%% file: main.tex
\def\draft{1}
\documentclass[11pt]{article}

\input{macros.tex}

\ifnum\draft=1
\newcommand{\mnote}[1]{{\color{red} [Madhu: #1]}}
\newcommand{\pnote}[1]{{\color{blue} [Prashanth: #1]}}
\newcommand{\snote}[1]{{\color{green} [Srikanth: #1]}}
\else  
\newcommand{\mnote}[1]{}
\newcommand{\pnote}[1]{}
\newcommand{\snote}[1]{}
\fi 

\date{\today}

\begin{document} 
	\title{Low-Degree Testing Over Grids\thanks{A conference version of this paper has appeared in RANDOM 2023.}}

        \ifnum\draft=1 
    \author{Prashanth Amireddy\thanks{School of Engineering and Applied Sciences, Harvard University, Cambridge, Massachusetts, USA. Supported in part by a Simons Investigator Award and NSF Award CCF 2152413 to Madhu Sudan and a Simons Investigator Award to Salil Vadhan. Email: \texttt{pamireddy@g.harvard.edu}} \and Srikanth Srinivasan\thanks{
     Department of Computer Science, University of Copenhagen, Denmark.This work was partially funded by the European Research Council (ERC) under grant agreement no. 101125652 (ALBA) and by a start-up package from Aarhus University. Some of this work was done while visiting the Meta-Complexity program at the Simons Institute for the Theory of Computing, UC Berkeley. Email: \texttt{srsr@di.ku.dk}} \and Madhu Sudan\thanks{School of Engineering and Applied Sciences, Harvard University, Cambridge, Massachusetts, USA. Supported in part by a Simons Investigator Award and NSF Award CCF 2152413. Email: \texttt{madhu@cs.harvard.edu}.}}
\else 
    \author{}
\fi 
	\maketitle
        \pagenumbering{arabic}
        \begin{abstract}
           We study the question of local testability of low (constant) degree functions from a product domain $\cS_1 \times \dots \times \cS_n$ to a field $\F$, where $\cS_i \subseteq \F$ can be arbitrary constant sized sets. We show that this family is locally testable when the grid is ``symmetric''. That is, if $\cS_i=\cS$ for all $i$, there is a probabilistic algorithm using constantly many queries that distinguishes whether $f$ has a polynomial representation of degree at most $d$ or is $\Omega(1)$-far from having this property. In contrast, we show that there exist asymmetric grids with $|\cS_1| = \cdots = |\cS_n| = 3$ for which testing requires $\omega_n(1)$ queries, thereby establishing that even in the context of polynomials, local testing depends on the structure of the domain and not just the distance of the underlying code. 
            
            The low-degree testing problem has been studied extensively over the years and a wide variety of tools have been applied to propose and analyze tests. Our work introduces yet another new connection in this rich field, by building low-degree tests out of tests for ``junta-degrees''. A function $f:\cS_1 \times \cdots \times \cS_n \to \cG$, for an abelian group $\cG$ is said to be a junta-degree-$d$ function if it is a sum of $d$-juntas. We derive our low-degree test by giving a new local test for junta-degree-$d$ functions. For the analysis of our tests, we deduce a small-set expansion theorem for spherical noise over large grids, which may be of independent interest. 
        \end{abstract} 
        
\newpage
        
\tableofcontents

\newpage 

\section{Introduction}\label{sec:intro}

The main problem considered in this paper is ``low-degree testing over grids''. Specifically given a degree parameter $d \in \Z^{\geq 0}$ and proximity parameter $\delta>0$ we would like to design a tester (a randomized oracle algorithm) that is given oracle access to a function  $f:\cS_1 \times \dots \times \cS_n \to \F$ where  
$\F$ is a field and $\cS_1,\ldots,\cS_n \subseteq \F$ are arbitrary finite sets, and accepts if $f$ is a polynomial of degree at most $d$ while rejecting with constant probability (say $1/2$) if $f$ is $\delta$-far (in relative Hamming distance) from every degree $d$ polynomial. The main goal here is to identify settings where the test makes $O(1)$ queries when $d,1/\delta$ and $\max_{i \in [n]}\{|\cS_i|\}$ are all considered constants. (In particular the goal is to get a query complexity independent of $n$.) 

\paragraph{Low-degree testing:} 
The low-degree testing problem over grids is a generalization of the classical low-degree testing problem which corresponds to the special case where $\F$ is a finite field and $\cS_1 = \cdots = \cS_n = \F$. Versions of the classical problem were studied in the early 90s \cite{BLR,BFL,BFLS} in the context of program checking and (multi-prover) interactive proofs. The problem was formally defined and systematically studied by Rubinfeld and Sudan~\cite{RubSud} and played a central role in the PCP theorem~\cite{AroraS,ALMSS} and subsequent improvements. While the initial exploration of low-degree testing focussed on the case where $d \ll |\F|$ (and tried to get bounds that depended polynomially, or even linearly, on $d$), a later series of works starting with that of Alon, Kaufman, Krivelevich, Litsyn and Ron~\cite{AKKLR} initiated the study of low degree testing in the setting where $d > |\F|$. \cite{AKKLR} studied the setting of $\F = \F_2$ and this was extended to the setting of other constant sized fields in \cite{JPRZ,KaufRon}. An even more recent sequence  of works ~\cite{BKSSZ,HSS,HRS,kaufman2022improved} explores so-called ``optimal tests'' for this setting and these results have led to new applications to the study of the Gowers uniformity norm, proofs of XOR lemmas for polynomials~\cite{BKSSZ}, and novel constructions of small set expanders~\cite{BGHMRS}. 

Part of the reason for the wide applicability of low-degree testing is the fact that evaluations of polynomials form error-correcting codes, a fact that dates back at least to the work of Ore~\cite{ore1922hohere}. Ore's theorem (a.k.a. the Schwartz-Zippel lemma) however applies widely to the evaluations of polynomial on entire ``grids'', i.e., sets of the form  $\cS_1 \times \dots \times \cS_n$ and bounds the distance between low-degree functions in terms of the degree $d$ and minimum set size $\min_i\{|\cS_i|\}$. This motivated Bafna, Srinivasan and Sudan~\cite{bafna2017local} to introduce the low-degree testing problem over grids. They proposed and analyzed a low-degree test for the special case of the Boolean grid, i.e., where $|\CS_1| = \cdots = |\CS_n| = 2$. This setting already captures the setting considered in \cite{AKKLR} while also including some novel settings such as testing the Fourier degree of Boolean functions (here the domain is $\{-1,+1\}^n$ while the range is $\R$). The main theorem in \cite{bafna2017local} shows that there is a tester with constant query complexity, thus qualitatively reproducing the theorem of \cite{AKKLR} (though with a worse query complexity than \cite{AKKLR} which was itself worse than the optimal result in \cite{BKSSZ}), while extending the result to many new settings. 

In this work we attempt to go beyond the restriction of a Boolean grid. We discuss our results in more detail shortly, but the main outcome of our exploration is that the problem takes on very different flavors depending on whether the grid is symmetric ($\cS_1 = \cdots = \cS_n$) or not. In the former case, we get constant complexity testers for constant $|\cS_i|$ whereas in the latter setting we show that even when $|\cS_i|=3$ low-degree (even $d=1$) testing requires superconstant query complexity. (See \cref{thm:deg-test} for details.) 
In contrast to previous testers, our tester goes via ``junta-degree-tests'', a concept that has  been explored in the literature but not as extensively as low-degree tests, and not been connected to low-degree tests in the past. We describe this problem and our results for this problem next.

\paragraph{Junta-degree testing:}

A function $f:\cS_1 \times \dots \times \cS_n \to \cG$ for an arbitrary set $\cG$ is said to be a $d$-junta if it depends only on $d$ of the $n$ variables. When $\cG$ is an abelian group, a function $f:\cS_1 \times \dots \times \cS_n \to \cG$ is said to be of junta-degree $d$ if it is the sum of $d$-juntas (where the sum is over $\cG$).\footnote{While in principle the problem could also be considered over non-abelian groups, in such a case it not clear if there is a fixed bound on the number of juntas that need to be summed to get to a function of bounded junta-degree.} 
In the special case where $|\cS_i|=2$ for all $i$ and $\cG$ is a field, junta degree coincides with the usual notion of degree. More generally every degree $d$ polynomial has junta degree $d$, while a function of junta-degree $d$ is a polynomial of degree at most $d\cdot \max_i\{|\cS_i|\}$. Thus junta-degree is softly related to algebraic degree and our work provides a step towards low-degree testing via the problem of junta-degree testing. 

Junta-degree testing considers the task of testing if a given function has junta-degree at most $d$ or if it is far from all functions of junta-degree at most $d$. While this problem has not been considered in full generality before, two works do consider this problem for the special case of $d=1$. Dinur and Golubev~\cite{dinur2019direct} considered this problem in the setting where $\cG = \F_2$, while Bogdanov and Prakriya~\cite{bogdanov2021direct} consider this for general abelian groups. This special case corresponds to the problem of testing if a function is a direct sum, thus relating to other interesting classes of properties studied in testing.  Both works give $O(1)$ query testers in their settings, but even the case of $d=2$ remained open. 

In our work we give testers for this problem for general constant $d$ in the general asymmetric domain setting with the range being an arbitrary finite group $\cG$, though with the restriction that the maximum set size $|\cS_i|$ is bounded. We then use this tester to design our low-degree test over symmetric grids. We turn to our results below. Even though our primary motivation in studying low-{\em junta}-degree testing is to ultimately use it for low-degree testing, we note that junta-degree testing even for the case of $\cG$ being the additive group of $\R$ (or $\C$) and $\cS_i = \Omega$ (which is some finite set) for all $i$, is by itself already interesting as in this case, junta-degree corresponds to the ``degree'' of the Fourier representation of the function (in any basis). Low-{\em Fourier}-degree functions and such approximations form a central object in complexity theory and computational learning theory, at least when the domain size is $|\Omega|=2$. The problem of {\em learning} low-Fourier-degree functions in particular has received much attention over the years~\cite{linialMN93,o2021analysis}, and hence {\em testing} the same family, over general domains $\Omega$, is an interesting corollary of our results, especially since our techniques are more algebraic than analytic (modulo the usage of a hypercontractivity theorem).


 \subsection{Our results}
 \label{sec:intro-results}

We start by stating our theorem for junta-degree testing. (For a formal definition of a tester, see~\cref{def:local-tester}). 


\begin{theorem}
\label{thm:gen-dom-junta}
    The family of junta-degree-$d$ functions from $\cS_1 \times \dots \times \cS_n$ to $\cG$ is locally testable with a non-adaptive one-sided tester that makes $O_{s,d}(1)$ queries to the function being tested, where $s=\max_i \abs{\cS_i}$.
    
    In the special case where $\abs{\cS_i} = s$ for all $i$,  the tester makes $s^{O(s^2d)}$ queries.
\end{theorem}

 In particular, if we treat all the parameters above except $n$ as constant, this gives a test that succeeds with high probability by making only a constant number of queries. 
 Taking $(\cG,+) = (\R,+)$ or $(\C,+)$, the above theorem results in a local tester for Fourier-degree:


\begin{corollary} The family of functions $f:\Omega^n \to \R$ of Fourier-degree at most $d$ is locally testable in $s^{O(s^2 d)}=O_{s,d}(1)$ queries, where $s=\abs{\Omega}$.\footnote{The same result also holds if the co-domain is $\C$ instead of $\R$.}
\end{corollary}
We now turn to the question of testing whether a given function $f:\cS^n \to \F$ is {\em degree-$d$}, i.e., whether there is a polynomial of degree at most $d$ agreeing with $f$, or $\delta$-far from it. Here $\cS$ can be any arbitrary finite subset of the field. 
Note that being junta-degree-$d$ is a necessary condition for $f$ being degree-$d$. Combining the above~\juntadegtest~with an additional test (called~\weakdegtest~), we can test low-degree functions over a field, or rather over {\em any subset} of a field.

\begin{restatable}{theorem}{degtestable}
\label{thm:deg-test}
    For any subset $\cS\subseteq \F$ of size $s$, the family of degree-$d$ functions from $\cS^n$ to $\F$ is locally testable with a non-adaptive, one-sided tester that makes ${(sd)}^{O(s^3 d)}=O_{s,d}(1)$ queries to the function being tested.
\end{restatable}

The special case of $\cS= \F = \F_q$ (finite field of size $q$) is especially interesting. Although this was already established for general finite fields first by~\cite{KaufRon} and an optimal query complexity (in terms of $d$, for constant prime $q$) was achieved in~\cite{HSS}, we nevertheless present it as a corollary of~\cref{thm:deg-test}. 

\begin{corollary}[\cite{KaufRon}]
    The family of degree-$d$ functions $f:\F_q^n \to \F_q$ is locally testable in $(qd)^{O(q^3 d)} = O_{q,d}(1)$ queries.
\end{corollary}

Turning our attention to more general product domains, we show that while junta-degree testing is still locally testable over there more general grids, local degree testing, even for $d=1$, is intractable for all sufficiently large fields.

\begin{theorem}
    \label{thm:impossible}
     For a growing parameter $n$, there exists a field $\F$ and its subsets $\cS_1,\dots,\cS_n$ of constant size (i.e., $3$) such that testing the family of degree-$1$ functions $f:\cS_1 \times \dots \times \cS_n\to \F$ requires $\Omega(\log n)$ queries.
\end{theorem}

\begin{remark}
    A recent work of Arora, Bhattacharyya, Fleming, Kelman and Yoshida~\cite{arora} considers low-degree testing over the reals and tests whether a given $f:\R^n \to \R$ is degree-$d$ or $\varepsilon$-far with respect to a distribution $\cD$. They give a test with query complexity independent of $n$ for their problem (\cite[Theorem 1.1]{arora}).
    This seems to contradict our result which seems to include the special case of their setting  for $\cD=\Unif(\cS_1 \times \dots \times \cS_n)$ and $\F=\R$, where \cref{thm:impossible} shows that a dependence on $n$ is necessary. The seeming contradiction is resolved by noting that the models in our paper and that of~\cite{arora} are quite different. In particular, while in our setting the function $f$ can only be queried on the support of the distribution $\cD$ (namely $\cS_1 \times \dots \times \cS_n$), in~\cite{arora} the function can be queried at any point in $\R^n$ and the distribution $\cD$ only shows up when defining the distance between two functions. (So in their model a function $f$ that happens to agree with a degree $d$ polynomial on the support of $\cD$ but disagrees outside the support may be rejected with positive probability, while in our model such a function must be accepted with probability one.)
\end{remark}

\subsection{Technical contributions}

All low-degree tests roughly follow the following pattern: Given a function $f$ on $n$ variables $x_1,\ldots,x_n$ they select some $k=O_d(1)$ new variables $y=(y_1,\ldots,y_k)$ and substitute $x_i = \sigma_i(y)$, where $\sigma_i$'s are simple random functions, to get an $O(1)$-variate function $g(y) = f(\sigma(y))$; and then verify $g$ is a low-degree polynomial in $k$ variables by brute force. When the domain is $\F_q^n$ for some field $\F_q$, $\sigma_i$'s can be chosen to be an affine form in $y$ --- this preserves the domain and ensures degree of $g$ is at most the degree of $f$, thus at least ensuring completeness. While soundness of the test was complex to analyze, a key ingredient in the analysis is that for any pair of points $a\ne b \in \F_q^k$,  $\sigma(a)$ and $\sigma(b)$ are uniform independent elements of $\F_q^n$ (over the randomness of $\sigma$). At least in the case where $f$ is roughly $1/q^k$ distance from the degree $d$ family, this ensures that with constant probability $g$ will differ from a degree $d$ polynomial in exactly one point making the test reject. Dealing with cases where $f$ is much further away is the more complex part that we won't get into here.

When the domain is not $\F_q^n$ affine substitutions no longer preserve the domain and so we can't use them in our tests. In the cases of the domain being $\{-1,+1\}^n$, \cite{bafna2017local} used much simpler affine substitutions of the form $x_i = c_i y_{j(i)}$ where $c_i \in \{-1,+1\}$ uniformly and independently over $i$ and $j(i) \in \{1,\ldots,k\}$ uniformly though not independently over $i$. Then \cite{bafna2017local} iteratively reduce the number of variables as follows: When only $r$ variables $x_1,\ldots,x_r$ remain, they pick two uniformly random indices $i\ne j \in \{1,\ldots,r-1\}$ and identify $x_j$ with $x_i$, and then rename the $r-1$ remaining variables as $x_1,\ldots,x_{r-1}$. At the end when $r=k$, they pick a random bijection between $x_1,\ldots,x_k$ and $y_1,\ldots,y_k$. This iterative identification eventually maps every variable $x_i$ to some variable $y_{j(i)}$. The nice feature of this identification scheme is it leads to a sequence of functions $f_n,f_{n-1},...,f_k$  with $f_r$ being a function of $r$ variables on the same domain, and of degree at most $d$ if $f=f_n$ has degree at most $d$. If however we start with $f_n$ being {\em very} far from degree $d$ polynomials, there must exist $r$ such that $f_{r}$ is very far from high-degree functions while $f_{r-1}$ is only {\em moderately} far. The probability of a bad event can be bounded (via some algebraic arguments) by $O(d^2/r^2)$. This step is the key to this argument and depends on the fact that $f_{r-1}$ involves very small changes to $f_r$. Summing over $r$ then gives the constant probability that the final function $f_k$ (or equivalently $g$) is far from degree $d$ polynomials. This still leaves \cite{bafna2017local} with the problem of dealing with functions $f$ that are close to codewords: Here they use the fact that this substitution ensures that $\sigma(a)$ is distributed uniformly in $\{-1,+1\}^n$ for every $a\in \{-1,+1\}^k$. It is however no longer true that $\sigma(b)$ is uniform conditioned on $\sigma(a)$ for $b\ne a$, but it is still the case that if $b$ is moderately far in Hamming distance from $a$  then $\sigma(b)$ has sufficient entropy conditioned on $\sigma(a)$. (Specifically $\sigma(b)$ is distributed uniformly on a sphere of distance $\Omega(n)$ from $\sigma(a)$.) This entropy, combined with appropriate small-set expansion bounds on the Boolean hypercube, and in particular a spherical hypercontractivity result due to Polyanskiy~\cite{polyanskiy2019hypercontractivity}), ensures that if $f$ is somewhat close to a low-degree polynomial then $g$ is far from every degree $d$ polynomial on an appropriately chosen subset of $\{-1,+1\}^k$ and so the test rejects.

To extend this algorithm and analysis to the setting on non-Boolean domains we are faced with two challenges: (1) We cannot afford to negate variables (using the random variables $c(i)$ above) when the domain is not $\{-1,+1\}$ -- we can only work with identification of variables (or something similar). (2) The increase in the domain size forces us to seek a general spherical hypercontractivity result on non-Boolean alphabets and this is not readily available. Overcoming either one of the restrictions on its own seems plausible, but doing it together (while also ensuring that the sequence of restrictions/identifications do not make the distance to the family being tested to abruptly drop in distance as we go from $f_n,f_{n-1},\dots$ to $f_k$) turns out to be challenging and this is where we find it critical to go via junta-degree testing.

As a first step in our proof we extend the approach of \cite{bafna2017local} to junta-degree testing over the domain $\CS^n$ for arbitrary finite $\CS$. (It is relatively simple to extend this further to the case of $\CS_1 \times \cdots \times \CS_n$ --- we don't discuss that here.) This is achieved by using substitutions of the form $x_i = \pi_i(y_{j(i)})$ where $\pi_i:\CS\to\CS$ is a random bijection. While this might increase the degree of the function, this preserves the junta-degree (or reduces it) and makes it suitable for analysis of the junta-degree test, which we now describe: Following the template of a low-degree test stated at the beginning of this subsection, the junta-tester would simply check whether $g(y) = f(\sigma(y))$ is of junta-degree at most $d$ where $\sigma$ is the random function induced by the identifications $j(.)$ and permutations $\pi_i$ of variables. The permutations $\pi_i$ here serve the same purpose as the coefficients $c_i$'s do in the substitutions $x_i=c_i y_{j(i)}$ of~\cite{bafna2017local} which is to ensure that for any $a\in \cS^k$, $\sigma(a)$ is uniformly distributed in $\cS^n$. With this idea in place extending the analysis of \cite{bafna2017local} to our setting ends up with a feasible path, except we had to address a few more differences; one such challenge is that in the analysis the rejection probability of junta-degree test on functions that are close to being junta-degree-$d$, we will need to analyze the effect of a spherical noise operator on grids (i.e., a subset of coordinates of fixed size is chosen uniformly at random and each coordinate in that subset is changed to a different value uniformly at random). While~\cite{polyanskiy2019hypercontractivity} shows that such a noise operator has the desired hypercontractivity behavior, and the corresponding small-set expansion theorem was used in the test of~\cite{bafna2017local}, this was only for a Boolean alphabet. In this paper, when the alphabet size $s=\abs{\cS}$ is more than $2$, by doing Fourier analysis over $\Z_s^n$, we are able to relate it to the more standard bernoulli i.i.d.~noise operator for which we do have a small-set expansion theorem available --- we believe this can be of independent interest.\footnote{We note that the hypercontracitivity setting we are considering and analyzing in this part is not sufficient to get a direct analysis of low-degree testing. Such an analysis would require hypercontractivity for more delicate noise models than the simpler ``$q$-ary symmetric'' models we analyze here.}

The other differences of our junta-degree test analysis compared to that of~\cite{bafna2017local} are mainly to account for the fact that we are aiming for junta-degree testing over any (abelian) group whereas the low-degree testing ideas of~\cite{bafna2017local} and other prior work utilize the properties of polynomials over fields. 
We also give a cleaner proof as compared to~\cite{bafna2017local} for the fact that the sequence of functions of fewer and fewer variables obtained by the random identifications (along with permutations) does not abruptly decrease in distance to the junta-degree-$d$ family like we pointed out earlier (see ``large-distance lemma''~\cref{lem:large}). 

We then return to the task of low-degree testing: For this we design a new test: We first test the given oracle for junta-degree $d$, then if it passes, we pick a fresh random identification scheme setting $x_i = y_{j(i)}$ for uniform independent $j(i) \in \{1,\ldots,k\}$ and verify (by brute-force) that the resulting $k$ variate function has degree at most $d$. The advantage with this two stage tester is that in the second stage the given function is already known to be close to a polynmial of degree at most $sd$ where $s = \max_i\{|\CS_i|\}$. This makes the testing problem closer to a polynomial identity testing problem, though the problem takes some care to define, and many careful details to be worked out in the analysis. 
A particular challenge arises from the fact that the first phase only proves that our function is only {\em close} to a low-degree polynomial and may not be low-degree exactly -- so in the second stage we have to be careful to sample the function on essentially uniform inputs. This prevents us from using all of $\CS^k$ when looking at the restricted function $g(y)$, but only allows us to use balanced inputs in $\CS^k$ (where a balanced input has an equal number of coordinates with each value $v \in \CS$). In turn understanding what the lowest degree of function can be given its values on a balanced set leads to new algebraic questions.  
\cref{sec:low-degree-test} gives a full proof of the low-degree test and analysis spelling out the many technical questions and our solutions to those. 



The final testing-related result we prove is an impossibility result, showing that while low-degree functions are locally testable over $\cS^n$, this cannot be extended to general grids $\cS_1 \times \dots \times \cS_n$. From a coding theory perspective, this reveals that local testability of even polynomial evaluations codes requires more structure than simply having a large distance.
 
\paragraph{Organization:} In \cref{sec:prelims} we define local testability and introduce {\em junta-polynomials} over an abelian group and observe that they have certain nice properties just like formal polynomials over a field. In~\cref{sec:low-junta-degree-test} we describe the junta-degree testing algorithm (\juntadegtest) over grids of the form $\Z_s^n$ and show its correctness by splitting the analysis into a ``small-distance lemma'' and a ``large-distance lemma''. Then in \cref{sec:low-degree-test} we use the ~\juntadegtest~test alongside an additional test (called~\weakdegtest) to give a low-degree test~\degtest~over grids of the form $\cS^n$ for any $\cS \subseteq \F$. We extend the~\juntadegtest~test from \cref{sec:low-junta-degree-test} to arbitrary product domains in \cref{sec:gen-dom}. In that section, we will also prove that the question of local low-degree testing over general grids $\cS_1 \times \dots \times \cS_n$ (where $\cS_i$'s are not necessarily identical) is intractable. 
Finally in \cref{sec:hypercontractive} we prove a small-set expansion theorem for spherical/hamming noise over grids, which would be used in the proof of the small-distance lemma of junta-degree testing.

\section{Preliminaries}
\label{sec:prelims}


We denote $[n]=\{1,\dots,n\} \subseteq \Z$, $[m..n] = [n] \setminus [m-1]$ and $\Z_s = \Z/s\Z = \{0,1,\dots,s-1\}$ for $s\ge 2$. Throughout the paper, let $(\cG,+)$ be an arbitrary abelian group and $(\F,+,\cdot)$ an arbitrary field. $\F_q^n$ is a vector space over the finite field of $q$ elements, to which we associate an inner product (bilinear form) as: $\innerprod{x,y} = \sum_{i=1}^n x_i \cdot y_i$.

For any finite set $\cS$ and $a \in \cS^n$ we denote the Hamming weight of $a$ by $\#a=\{i\in [n]:a_i \ne 0\}$, assuming $\cS$ contains an element called $0$. 
If $I \subseteq [n]$, we use $a^I$ to denote the tuple $a$ restricted to the coordinates of $I$, i.e., $a^I = (a_i)_{i\in I}$.  Similarly $\cS^I = \{a^I:a\in \cS^n\}$. For disjoint subsets $I,J\subseteq [n]$, and $a \in \cS^I$ and $b \in \cS^J$, we denote their concatenation by $a \circ b \in \cS^{I \cup J}$. Denoting a product domain/grid by $\overline{\cS} = \cS_1\times \dots\times \cS_n$, we let $\overline{\cS}^I = \times_{i\in I} \cS_i$ denote the Cartesian product of sets restricted to the coordinates of $I$.

We use $\binom{[n]}{\le d}$ to denote the set of subsets of $[n]$ of size at most $d$. For $m$ a multiple of $s$, let ``balanced set'' $\cB(\cS,m) \subseteq \cS^m$ be the set of points that contain exactly $m/s$ many repetitions of each element of $\cS$. Abusing notation, sometimes we may think of $\cB(\cS,m)^{m'}$ as a subset of $\cS^{mm'}$ by flattening the tuple of $m$-tuples.

\begin{definition}[{\bf Group-integer multiplication}]
    The {\em group-integer multiplication} operation $\cdot:\cG \times \Z \to \cG$ is defined as 
    $$g \cdot m = 
    \begin{cases}
        \case \underbrace{g + \dots + g}_{\abs{m}\text{~times}} \text{~if~} m \ge 0\\
        \case \underbrace{-g - \dots - g}_{\abs{m}\text{~times}} \text{~otherwise}.
    \end{cases}$$ Note that both the group addition and integer addition distribute over this operation i.e., $$(g_1 + g_2)\cdot (m_1+m_2) = g_1\cdot m_1 + g_1 \cdot m_2 + g_2\cdot m_1 + g_2\cdot m_2.$$ 
\end{definition}


\subsection{Local testability} 

We start by defining the notion of a distance to a family of functions.

\begin{definition}[{\bf $\delta$-far}] The {\em distance} between $f: \overline{\cS} \to \cG$ and a family of functions $\cF$ with the same domain $\overline{\cS}$ is $$\delta_\cF(f) = \min_{g \in \cF} \delta(f,g),$$
where $$\delta({f,g})=\Pr_{x \sim \overline{\cS}} \brac{f(x) \ne g(x)}.$$ We say that $f$ is {\em $\delta$-far} from $\cF$ if $\delta_\cF(f) \ge \delta$. When $\cF$ is the family of junta-degree-$d$ functions, we denote $\delta_{\cF}(.)$ by simply $\delta_{d}(.)$. Similarly $f$ is {\em $\delta$-close} to $\cF$ if $\delta_{\cF}(f) \le \delta$.
\end{definition}

\begin{definition}[{\bf Local testability}]\label{def:local-tester}

A randomized algorithm $\cA$ with an oracle access to a function $f:\overline{\cS} \to \cG$ as its input, is said to be {\em $q$-local} if it performs at most $q$ queries for any given $f$. For a family of functions $\cF$ with domain $\overline{\cS}$ and co-domain $\cG$, we say that $\cF$ is {\em $q$-locally testable} for $q=q(\cF)$ if there exists a $q$-local test $\cA$ that accepts $f$ with probability $1$ if $f\in \cF$, and rejects $f$ with probability at least $\delta_{\cF}(f)/2$ if $f\notin \cF$. 
Further if $q(\cF) = O(1)$, we simply refer to $\cF$ as being {\em locally testable}.

We say that an algorithm $\cA$ is a {\em one-sided} test for $\cF$ if it always accepts if $f\in \cF$ and that $\cA$ is {\em non-adaptive} if all the queries are predetermined (perhaps according to a distribution) and the result of $\cA$ is a deterministic predicate of the outputs of those queries.
\end{definition} 

The family of functions $\cF$ of our study (namely ``junta-degree-$d$'' and ``degree-$d$'' to be formally defined shortly) are parameterized by $s=\max_i \abs{\cS_i}$ and an integer $d$ which we treat as constants. All tests we are going to present are $O_{s,d}(1)$-local, one-sided and non-adaptive. However, the probability of rejection in case of $f\notin \cF$ is only $\Omega_{s,d}(\delta_{\cF}(f))$; nevertheless by repeating the test an appropriate $O_{s,d}(1)$ number of times, we get a $O_{s,d}(1)$-local test for $\cF$ that succeeds with probability $\delta_\cF(f)/2$ when $f\notin \cF$ and with probability $1$ when $f\in \cF$. 
In the context of this paper, the above definition for local testability is without loss of generality as we know from~\cite{Ben-SassonHR05} that for {\em linear properties}\footnote{i.e., for families $\cF$ for which $f\in \cF$ and $g\in \cF$ implies $c_1 f+ c_2 g \in \cF$ for all $c_1,c_2\in \F$.} (applicable when the co-domain is a field), any ``test'' can be transformed to be one-sided and non-adaptive without altering the query complexity (locality) and success probability by more than constant factors.

\subsection{Junta-polynomials and polynomials}
\label{sec:junta-polys}

For this section, we let $\overline{\cS}=\cS_1\times \dots\times \cS_n$ denote an arbitrary finite product domain (or grid) and $s=\max_i \set{s_i}$, where $s_i = \abs{\cS_i}$.

\begin{definition}[{\bf Junta-degree}]
    A function $f:\overline{\cS} \to \cG$ \footnote{Here we treat a tuple of sets as the domain of the function} is said to be {\em junta-degree-$d$} if $$f(x)=f_1(x^{D_1}) + \dots + f_t(x^{D_t})$$ for some  $t\in \Z,D_j \in \binom{[n]}{\le d}$ and functions $f_j:\overline{\cS}^{D_j} \to \cG$ for $j\le t$. If $t=1$, we call $f$ a {\em $d$-junta}.
    
    The {\em junta-degree} of $f$ is the minimum $d \ge 0$ such that $f$ is junta-degree-$d$.
\end{definition}

For junta-degree testing over arbitrary grids $\overline{\cS} = \cS_1 \times \dots \times \cS_n$, we may assume that $\cS_i = \Z_{s_i}$ without loss of generality, where $s_i = \abs{\cS_i}$. We prove the following claims about junta-polynomials; these are analogous to standard facts about multi-variate polynomials over a field.

\begin{claim}
\label{clm:junta-unique}
    Any junta-degree-$d$ function $f:\Z_{s}^n\to \cG$ can be uniquely\footnote{up to the commutativity of the $\Sigma$ (group addition) and $\Pi$ (integer multiplication) operations} expressed as 
    \begin{align}\label{eqn:junta-poly}
    f(x_1,\dots,x_n) = \sum_{\substack{a \in {\Z_s}^n \\ \#a \le d}} g_{a} \cdot \prod_{i\in [n]:~a_i \ne 0} \delta_{a_i}(x_i),\end{align} where $g_{a}\in \cG$ and $\delta_{b}:\Z_s \to \Z$ is defined as $\delta_{b}(y)=1$ if $b=y$ and $0$ otherwise.
    
\end{claim}

\begin{definition}[{\bf Junta-polynomial}]
    We will call such a representation as a {\em junta-polynomial}, and the {\em degree} of a junta-polynomial is defined as $\max_{a \in {\Z_s}^n:g_a \ne 0} \#a$. It can be seen that the degree of a junta-polynomial is exactly equal to the junta-degree of the function it computes, assuming that the degree of the identically $0$ junta-polynomial is $0$.
    
    We will refer to the summands in~\eqref{eqn:junta-poly} as {\em terms}, the constants $g_a$ as {\em coefficients}, the integer products $\prod_{i\in [n]:~a_i \ne 0} \delta_{a_i}(x_i)$ as {\em monomials}. We say that $a$ is a {\em root} of a junta-polynomial $P$ if $P(a)=0$ and $a$ is a {\em non-root} otherwise.
\end{definition}

\begin{proof}[Proof of~\cref{clm:junta-unique}]
    Let $f(x_1,\dots,x_n) = f_1(x^{D_1}) + \dots f_t(x^{D_t})$ for some $D_1,\dots,D_t \subseteq [n]$ of size at most $d$. Then by viewing $f_j$ in terms of its ``truth table'', we have
    \begin{align*}  
        f_j(x^{D_j}) & = \sum_{a \in \Z_s^{D_j}} f_j(a) \cdot \prod_{i \in D_j} \delta_{a_i} (x_i)\\
        & = \sum_{a \in \Z_s^{D_j}} f_j(a) \cdot \prod_{i \in D_j:~a_i \ne 0} \delta_{a_i} (x_i) \prod_{i \in D_j:~a_i = 0} \delta_{0} (x_i) \\
        & = \sum_{a \in \Z_s^{D_j}} f_j(a) \cdot \prod_{i \in D_j:~a_i \ne 0} \delta_{a_i} (x_i) \prod_{i \in D_j:~a_i = 0} \paren{1-\sum_{b \in \Z_s\setminus \{0\}} \delta_b (x_i)}
    \end{align*}
    Thus, by expanding the above expression and adding up the junta-polynomials of $f_j$ across $j\in [t]$, we end up with a junta-polynomial for $f$. Its degree is at most $d$ as for any $j$ and $a\in \Z_s^{D_j}$, the number of $\delta(.)$ factors in each monomial is at most $\abs{\set{i\in D_j:a_i\ne 0}} + \abs{\set{i\in D_j:a_i=0}}=\abs{D_j} \le d$.

    To prove uniqueness, for any function that has a non-zero junta-polynomial representation, we will show that it has at least one non-root. In fact, we will prove something stronger in the following claim: that it has at least $s^{n-d}$ non-roots.
\end{proof}

\begin{claim}
\label{clm:junta-polys-sz}
    Any non-zero junta-polynomial $P:\Z_s^n \to \cG$ of degree at most $d$  has at least $s^{n-d}$ non-roots.
\end{claim}

\begin{proof}
    The proof is by induction on $n$. Let $$P(x) = \sum_{\substack{a \in \Z_s^n \\ \#a \le d}} g_{a} \cdot \prod_{i\in [n]:~a_i \ne 0} \delta_{a_i}(x_i)$$ where at least one coefficient $g_a$ is non-zero.
    
    \noindent {\bf Base case: $n=1$.} In this case, $P$ is of the form
    \begin{align}
        P(x_1) = g_0 \cdot 1 + g_1 \cdot \delta_{1}(x_1) + \dots + g_{s-1} \cdot \delta_{s-1}(x_1).
    \end{align}
    \begin{itemize}
        \item If $g_i=0$ for all $i > 0$, then $g_0 \ne 0$ and $P(x_1) = g_0 \ne 0$ for all $x_1 \in \s$, so it has $s=s^{1-0}=s^{n-d}$ non-roots.
        \item If $g_i \ne 0$ for some $i > 0$ and $g_0=0$, then $d=1$ and $P(i) = g_0 + g_i \ne 0$. Thus, $P$ has at least $1=s^{n-d}$ non-root.
        \item If $g_i \ne 0$ for some $i > 0$ and $g_0 \ne 0$, then $d=1$ and $P(0) = g_0 \ne 0$ and $P$ has at least $1=s^{n-d}$ non-root.
    \end{itemize}
    
    \noindent {\bf Induction step: $n\ge 2$.} By distributivity of the group-integer multiplication, we have
    \begin{align}
        P(x_1,\dots,x_{n-1},x_n) = P_0(x_1,\dots,x_{n-1})\cdot 1 + P_1(x_1,\dots,x_{n-1}) \cdot \delta_{1}(x_n) + \dots + P_{s-1}(x_1,\dots,x_{n-1}) \cdot \delta_{s-1}(x_n)
    \end{align}
    for some junta-polynomials $P_0,P_1,\dots,P_{s-1}$ over $n-1$ variables. Since $P$ is a non-zero junta-polynomial, we have the following cases.
    \begin{itemize}
        \item If $P_i=0$ (i.e., identically zero junta-polynomial) for all $i>0$, then $P_0$ is a non-zero junta-polynomial (of degree at most $d$), so by induction hypothesis it must have at least $s^{n-1-d}$ non-roots. Notice that if $P_0(x^*_1,\dots,x^*_{n-1})\ne 0$, then $P(x^*_1,\dots,x^*_{n-1},x^*_n) = P_0(x^*_1,\dots,x^*_{n-1}) + 0 + \dots + 0 \ne 0$ for all $x^*_n \in \Z_s$. Hence, $P$ has at least $s \cdot s^{n-1-d} = s^{n-d}$ non-roots.
        \item If $P_i \ne 0$ for some $i > 0$, since the degree of $P$ is at most $d$, the degree of $P_i$ is at most $d-1 \ge 0$. Hence, $P_i$ has at least $s^{n-1-(d-1)} = s^{n-d}$ non-roots. For each such non-root $(x^*_1,\dots,x^*_{n-1})$, we argue that there exists $x^*_n \in \s$ such that $(x^*_1,\dots,x^*_{n-1},x^*_n)$ is a non-root for $P$: this follows because 
        $$P(x^*_1,\dots,x^*_{n-1},x_n) = P_0(x^*_1,\dots,x^*_{n-1}) + \sum_{1\le j  \le s-1} P_j(x^*_1,\dots,x^*_{n-1}) \cdot \delta_j (x_n)$$ when treated as a junta-polynomial over the variable $x_n$ is non-zero since the coefficient $P_i(x^*_1,\dots,x^*_{n-1}) \ne 0$. Therefore, there must be at least one choice of $x_n=x^*_n \in \s$ such that $P(x^*_1,\dots,x^*_{n-1},x^*_n) \ne 0$ (by the base case of the induction). Hence, $P$ also has at least $s^{n-d}$ non-roots.
    \end{itemize}
\end{proof}

We will now discuss standard facts about formal polynomials.
\label{subsec:polys}
Let $\F$ be a field and $\cS \subseteq \F$ be of size $s\ge 2$. For a polynomial $P(x_1,\dots,x_n) \in \F[x_1,\dots,x_n]$ the individual degree of $x_i$ is the largest degree $x_i$ takes in any (non-zero) monomial of $P$. The individual degree of $P$ is the largest individual degree of any variable $x_i$. We say that $P$ is {\em degree-$d$} if its degree is at most $d$.
We say that $f:\cS^n \to \F$ is degree-$d$ iff there is a degree-$d$ polynomial $P\in \F[x_1,\dots,x_n]$ computing $f$. For the analysis of our degree-tester, we also need a notion of degree-$d$ for non-product domains: for any $\cT \subseteq \cS^n$, we say that $f:\cT \to \F$ is degree-$d$ if there is a degree-$d$ polynomial $P\in \F[x_1,\dots,x_n]$ computing $f$. 

\begin{claim}
\label{clm:unique-poly}
    Any degree-$d$ function $f:\cS^n \to \F$ has a unique polynomial representation with degree at most $d$ and individual degree at most $s-1$.
\end{claim}
 
By setting $d=n(s-1)$ (or $\infty$) in the above claim, we see that the set of all functions from $\cS^n$ to $\F$ is a vector space over $\F$ of dimension $s^n$ -- the monomials with individual degree at most $s-1$ form a basis. More generally, for any $\cT \subseteq \cS^n$\footnote{For example, when $\cT=\cS_1 \times \dots \times \cS_n$} the set of functions from $\cT$ to $\F$ forms a vector space of dimension $\abs{\cT}$ with an inner product defined for $f,g:\cT \to \F$ as $\innerprod{f,g} = \sum_{x\in \cT} f(x) \cdot g(x)$. For any $d$, the set of degree-$d$ functions is a subspace of this vector space.

It is easy to see that if $f:\cS^n \to \F$ is degree-$d$, then it is also junta-degree-$d$ (w.r.t.~to the additive group of $\F$). Conversely, if $f:\cS^n \to \F$ is junta-degree-$d$, then it is degree-$(s-1)d$: this follows by applying~\cref{clm:unique-poly} to the $d$-junta components of $f$. If $s=2$, the degree is exactly equal to the junta-degree.

Let $\delta'_d(f)$ denote the distance of $f$ to the degree-$d$ family. 

\subsection{Fourier analysis}
\label{sec:fourier}

\begin{definition}[{\bf Fourier representation}]
\label{def:fourier}
    Any function $f:\Z_s^n \to \C$ can be uniquely expressed as 
\begin{align}
    f(x) = \sum_{\alpha \in \Z_s^n} \widehat{f}(\alpha) \chi_\alpha(x)
\end{align}
where the {\em characters} are defined as $\chi_\alpha(x) = \prod_{i \in [n]} \chi_{\alpha_i}(x_i)$ where $\chi_{\beta} (y) = \omega^{\beta y\mod s}$ for $\beta,y \in \Z_s$ and $\omega \in \C$ is a (fixed) primitive $s$-th root of unity.
\end{definition}

\begin{claim}[{\bf Properties of characters}] We have $\chi_{0^n}(x) = \chi_\alpha(0^n) = 1$ for all $x,\alpha\in \Z_s^n$. Additionally,
    \begin{itemize} 
        \item For $\beta\in \Z_s$, $\E_{x\sim \Z_s} \brac{\chi_\beta(x)} = 1$, if $\beta=0$ and $0$ otherwise.
        \item For $x,\alpha,\beta \in \Z_s^n$, $\chi_{\alpha + \beta} (x) = \chi_\alpha(x)\chi_\beta(x)$
        \item For $\alpha,x,y \in \Z^n_s$, $\chi_\alpha(x+y) = \chi_\alpha(x)\chi_\alpha(y)$.
    \end{itemize}
\end{claim}

\begin{definition}[{\bf Noise}]
    For $\nu \in [0,1]$ and $x\in \Z_s^n$, we define $N_\nu(x)$ \footnote{This is different from the standard usage $N_\rho$ where $\rho$ denotes the probability of ``retention'' and not of noise.} to be the distribution over $\Z_s^n$ where each coordinate of $x$ is unchanged with probability $1-\nu$, and changes to a {\em different} value uniformly at random with probability $\nu$. Similarly, the {\em spherical noise} corresponds to $S_\nu(x)$ where a subset $J \subseteq [n]$ of fixed size $\nu n$ is chosen uniformly at random and the coordinates outside $J$ are unchanged and those within $J$ are changed to a uniformly different value.
    Let $\cD_\nu$ denote the probability distribution over $\Z_s$ with mass $1-\nu$ at 0 and $\nu/(s-1)$ at all the other points. Let $\cE_\nu$ denote the uniform distribution over $\set{y\in \Z_s^n: \#y = \nu n}$. For $\mu_1 \sim \cD_\nu^{\otimes n}$ and $\mu_2 \sim \cE_\nu$, note that $N_\nu(x)$ and $x+\mu_1$ are identically distributed; so are $S_\nu(x)$ and $x+\mu_2$.

\end{definition}

\begin{lemma}[{\bf Noise stability}]
\label{lem:ns}
    For any function $f:\Z_s^n \to \C$ with range real numbers and any random variable $\mu$ over $\Z_s^n$ (independent from $x\sim 
 \Z^n_s$), we have
    \begin{align}
        \E_{\substack{x\sim\Z^n_s\\ \mu}} \brac{f(x)f(x+\mu)} = \sum_{\alpha \in \Z^n_s} \abs{\widehat{f}(\alpha)}^2 \E_\mu\brac{\chi_\alpha(\mu)}.
    \end{align}
    Moreover, for any $\rho \in [0,1]$ and $\nu=(1-1/s)(1-\rho)$, if $\mu$ is $\cD_\nu^{\otimes n}$, then the inner factor $\E_\mu \brac{\chi_\alpha(\mu)}$ is equal to $\rho^{\#\alpha}$.
\end{lemma}

We define the notion of Fourier-degree of a function.

\begin{definition}[{\bf Fourier-degree}] 
    For a finite set $\Omega$, suppose a function $f:\Omega^n \to \R$ (or $\C$) has a Fourier representation as follows:
        \begin{align}
            f(x) = \sum_{\alpha \in \Omega^n} \widehat{f}(\alpha) \chi_\alpha(x),
        \end{align} 
        where $\chi_{\alpha}(x) = \prod_{i\in [n]} \chi_{\alpha_i}(x_i)$ and $\{\chi_\beta\}_{\beta \in \Omega}$ is an orthonormal basis for functions $\Omega \to \R$ (or $\C$), with $\chi_0 \equiv 1$ where an arbitrary element of $\Omega$ can be treated as $0$.
    We then define the Fourier-degree to be $\max\{\#\alpha : \widehat{f}(\alpha) \ne 0\}$, or $0$ if $f\equiv 0$.
\end{definition}

It is not hard to show that Fourier-degree of $f$ is exactly equal to the junta-degree of $f$ (with the co-domain being the additive group of $\R$ or $\C$), irrespective of the choice of the basis $\{\chi_\beta\}_{\beta \in \Omega}$ used in the above definition. 

\section{Low-junta-degree testing}
\label{sec:low-junta-degree-test}

We note that junta-degree-$d$ functions with domain $\cS_1 \times \dots \times \cS_n$ such that $\abs{\cS_i}=s$ for all $i$ are ``equivalent'' to those with domain $\Z_s^n$ as one can fix an arbitrary ordering of elements in each $\cS_i$ and treat the function as being over $\Z_s^n$: this does not change the junta-degree. Hence, we will fix $\cS_i = \Z_s$. The more general case of unequal domain sizes will be handled in~\cref{sec:gen-dom}.

We claim that the following test works to check if a given function $f:\Z_s^n \to \cG$ is junta-degree-$d$. 

\paragraph{The junta-degree test~(\juntadegtest):}
\label{sec:overview-junta-degree}
 For a parameter $k=O_{s,d}(1)$ that is yet to be fixed, the junta-degree test (which we shall refer to as~\juntadegtest) for $f:\Z_s^n \to \cG$ is the following algorithm with $I=[n]$, $r=n$ and $f_r = f$:

\begin{mdframed}
\label{alg:1}
    \noindent {\bf Test $T_{I,r}(f_r)$: gets query access to $f_r:\Z_s^I \to \cG$ where $I \subseteq [n]$ is of size $r$}
    \begin{enumerate}
    \item \label{step:bf} If $r \le k$, accept iff $f_r$ is junta-degree-$d$ (check this by querying $f_r$ at all points in its domain). Otherwise,
    \item Choose $i \ne j \in I$ and a permutation $\pi_j:\Z_s \to \Z_s$ independently and uniformly at random. Let $I'=I\setminus \{j\}$.
    \item \label{step:rn} Apply the test $T_{I',r-1}(f_{r-1})$ where $f_{r-1}:\Z_s^{I'} \to \cG$ is the function obtained by setting $x_j = \pi_j(x_i)$ in $f_r$: that is, $f_{r-1}(a^{I'}) := f_r(a^{I'} \circ (\pi_j(a_i))^{\{j\}})$ for $a\in \cS^{I'}$.
    \end{enumerate}
\end{mdframed}

 The query complexity of the~\juntadegtest~test is $s^{k} = O_{s,d}(1)$ regardless of the randomness within the test. Furthermore, if the function $f$ happens to be a junta-degree-$d$ function, then the test~\juntadegtest~always accepts it, since permuting variables and substituting some  variables with other variables does not change the junta-degree, so Step~\ref{step:bf} succeeds. In this section, we will show that if $\delta:=\delta_d(f) > 0$, then $\Pr[\juntadegtest\text{ rejects }f] \ge \varepsilon\delta$ for appropriate $\varepsilon = \Omega_{s,d}(1)$.

We follow the same approach as~\cite{bafna2017local} (which itself follows~\cite{BKSSZ}) and argue that if $\delta_d(f)$ is ``small'', then we will be able to prove $\Pr[\juntadegtest\text{ rejects }f] \ge \varepsilon\!\cdot\!\delta_d(f)$ and if not, at least we will be able to find some $r\in [k+1..n]$ such that $\delta_d(f_r)$ is small enough (but importantly, not too small). Then, we apply the small-distance analysis to that $f_r$. 

We state here the two main lemmas to prove that the correctness of the junta-degree tester. Here, the parameters $\varepsilon_0 \le \varepsilon_1$ and $\varepsilon$ will be chosen to be at least $s^{-O(k)}$. In the context of the test $T_{I,r}(f_r)$ described above, we will set $k=\psi s^2d$ for a sufficiently large but constant $\psi$ to be fixed in the proofs of the below lemmas\footnote{For $d=0$, we can take $k=\psi s^2$ so that it is non-zero.}.

\begin{lemma}[{\bf Small-distance lemma}]
    \label{lem:small}
    For any $I\subseteq [n]$ of size $r>k$, if $\delta = \delta_d(f_r) \le \varepsilon_1$, then 
    $$\Pr[T_{I,r}\text{ rejects }f_r] \ge \varepsilon\delta.$$
\end{lemma}

\begin{lemma}[{\bf Large-distance lemma}]
    \label{lem:large}
    For any $I\subseteq [n]$ of size $r>k$, if $\delta_d(f_r) > \varepsilon_1$, then 
    $$\Pr_{i,j,\pi_j}[\delta_d(f_{r-1}) \le \varepsilon_0] \le {k^2}/2r(r-1).$$
\end{lemma}

Assuming the above two lemmas to be true, we will prove our main theorem now; we will prove a slightly different version below, which implies the statement of~\cref{thm:gen-dom-junta} (for symmetric domains $\cS_1 \times \dots \times \cS_n = \Z_s^n$) by simply repeating the ~\juntadegtest~test $O(1/\varepsilon)=O_{s,d}(1)$ times independently and taking an AND vote.

\begin{theorem} For any function $f:\Z_s^n \to \cG$, if $f$ is junta-degree-$d$ then \juntadegtest~always accepts $f$. Otherwise,
    $$\Pr[\juntadegtest~\text{rejects~}f] \ge \varepsilon \delta$$ for some $\varepsilon=s^{-O(s^2d)}$, where $\delta=\delta_d(f)$ is the distance to the junta-degree-$d$ family.
\end{theorem}
    
\begin{proof}
    If $f:\Z_s^n \to \cG$ is junta-degree-$d$,~\juntadegtest~always succeeds. Otherwise, if $\delta:=\delta_d(f) \le \varepsilon_1$, since $\delta_d(f) = \delta_d(f_n)$, by the small-distance lemma, we get 
    $$\Pr[\juntadegtest\text{ rejects }f] = \Pr[T_{[n],n}\text{ rejects }f_n] \ge \varepsilon\delta_d(f_n) = \varepsilon\delta.$$
    Hence, the remaining case to analyze is when $\delta > \varepsilon_1$.
    
    For any $r\in [k..n]$, denote the event that $\delta_d(f_r) < \varepsilon_0$ by $\cL_r$. Similarly use $\cM_r$ and $\cH_r$ to denote that the above distance is in the interval $[\varepsilon_0, \varepsilon_1] $ and that it is greater than $\varepsilon_1$ respectively. These events depend on the randomness $i,j,\pi_j$ that the test uses over different $r$. Notice that $\delta_d(f_n) = \delta > \varepsilon_1$, so $\Pr[\cH_n]=1$.

    \noindent We will drop $I$ from $T_{I,r}$ and let $T_r$ denote the test corresponding to the recursive call of the test $T_{[n],n}$ with $r$ variables. Since the test~\juntadegtest~rejects if $\cH_k$ occurs (i.e., if $\delta_d(f_k) > \varepsilon_1 > 0$),
    \begin{align*}
        \Pr\brac{T_{n}\text{ rejects }f_n \middle\vert\ \bigwedge_{r=k}^{n} \cH_r} & = 1.
    \end{align*}

    \noindent Now suppose $\cM_{r^*}$ occurs for some ${r^*} \ge k$, meaning $\delta_d(f_{r^*}) \in [\varepsilon_0, \varepsilon_1]$. If ${r^*}=k$, the test always rejects as $\varepsilon_0 > 0$. Hence suppose ${r^*} >k$. Conditioned on $\cM_{r^*}$ occuring, the probability of $T_{r^*}$ rejecting $f_{r^*}$ is at least $\varepsilon\delta_{d}(f_{r^*}) \ge \varepsilon\varepsilon_0$ by the small-distance lemma. Therefore
    \begin{align*}
        \Pr \brac{T_n\text{ rejects }f_n \middle\vert\ \bigvee_{r=k}^{n} \cM_{r} } \ge \varepsilon\varepsilon_0.
    \end{align*}

    \noindent Now, it suffices to show that $$\Pr\brac{\bigwedge_{r=k}^n \cH_r \vee \bigvee_{r=k}^n \cM_r} \ge \widetilde{\varepsilon}/\varepsilon\varepsilon_0,$$ in order to be able to conclude that 
    $$\Pr[T_n\text{ rejects} f] \ge \min\{\varepsilon\delta, \widetilde{\varepsilon}\}\ge \paren{\varepsilon\widetilde{\varepsilon}} \cdot \delta.$$ This is where we will use the large-distance lemma: it says for all $r > k$,
    \begin{align*}
        \Pr \brac{\cL_{r-1} \mid \cH_r} \le k^2/2r(r-1)~{\text{ or equivalently, } \Pr\brac{\cM_{r-1} \vee \cH_{r-1} \mid \cH_r} \ge 1-k^2/2r(r-1)}.
    \end{align*}

    \noindent Thus by repeatedly conditioning on $\cH_n,\dots,\cH_k$, and using the fact that the test uses independent randomness for each $r$, we obtain
    \begin{align*}
        \Pr\brac{\bigwedge_{r=k}^n \cH_r \vee \bigvee_{r=k}^n \cM_r} & \ge \prod_{r=k+1}^{n} \paren{1-k^2/2r(r-1)} \ge 4^{-k^2{\sum_{r=k+1}^n \frac{1}{2r(r-1)}}} \ge \widetilde{\varepsilon}/\varepsilon\varepsilon_0. 
    \end{align*}
    For the third inequality, we used the fact that $\sum_{r=k+1}^n\frac{1}{2r(r-1)} = O\left(\frac{1}{k}\right)$.
     We set $\widetilde{\varepsilon} := \varepsilon\varepsilon_0 2^{-O(k)}\ge s^{-O(k)}$ for the last inequality. 
\end{proof}

\subsection{Small-distance lemma}
\label{sec:small-distance}

\begin{proof}[Proof of~\cref{lem:small}]

We will ``unroll'' the recursion of the~\juntadegtest~test and state it more directly as follows:

\subsubsection{Same test, rephrased}
\label{sec:rephrased}

Fix an arbitrary $r>k$. 
As $r$ is fixed, we denote $f_r$ by $f$ (not to be confused with the initial function on $n$ variables). For the proof, we will need the following alternate description of $T_{I,r}$:

\begin{mdframed}
        \noindent {\bf Test $T_{I,r}(f)$: gets query access to $f:\Z_s^I \to \cG$ with $r$ many variables $(x_i)_{i\in I}$}
    \begin{enumerate}
    \item \label{step:1} Pick a tuple of permutations $\pi=(\pi_1,\dots,\pi_r)$ of $\Z_s$ uniformly and independently at random.
    \item \label{step:1.5} Pick a bijection $\mu:[r] \to I$ u.a.r.
    \item \label{step:2} Construct a map $p:\{k+1,\dots,r\}\to \Z$ as follows: For $i=k+1$ to $r$, 
    \begin{itemize}
    \item choose $p(i)$ to be an element from $[i-1]$ uniformly and independently at random.
    \end{itemize}
    \item \label{step:3} For $i=r$ to $k+1$ (in this order), 
    \begin{itemize}
        \item substitute $x_{\mu(i)}$ by $\pi_i(x_{{p(\mu(i))}})$.
    \end{itemize}
    \item For $i=1$ to $k$,
    \begin{itemize}
        \item replace $x_{\mu(i)}$ by $\pi_i(y_{{i}})$, where $y_1,\dots,y_k$ are new variables taking values from $\Z_s$.
    \end{itemize}
    \item Accept iff the restricted function of $f$, say $f'(y_1,\dots,y_k)$, is junta-degree-$d$.
    \end{enumerate}
\end{mdframed}

To justify that the above test is indeed equivalent to the original formulation, we note that the first 4 steps above are obtained by simply ``unrolling'' the recursion and step 5 does not change the rejection probability of the test (but will help in the analysis). Further, the random process that identifies the variables in step 4 can be understood in terms of the following distribution over maps $\sigma:[r]\to [k]$.

\begin{itemize}
\item For $i=1$ to $k$, set $\sigma(i)=i$.
\item For $i=k+1$ to $r$, set $\sigma(i)=j$ with probability $\abs{\set{i'<i:\sigma(i')=j}}/(i-1)$, for each $j\in [k]$.
\end{itemize}

The only property we need about the above distribution of $\sigma$ is that it is ``well-spread'', which was already shown in~\cite{bafna2017local} as the following lemma.
\begin{lemma}[Corollary 6.9 in~\cite{bafna2017local}]
\label{lem:spread}
    With probability at least $1/2^{O(k)}$, we have $\abs{\sigma^{-1}(j)} \ge r/4k$ for all $j\in [k]$ -- we call such a $\sigma$ {\em good}.
\end{lemma}

We now present yet another equivalent description of the test assuming $I=[r]$ -- this does not affect the probability of the test rejecting as we are only relabeling variables and step 2 above randomizes the variables anyway. We shall also drop the subscript $I$ in $T_{I,r}$.

\begin{mdframed}
        \noindent {\bf Test $T_r(f)$: gets query access to $f:\s^{[r]} \to \cG$ with variables $x_1,\dots,x_r$}
    \begin{enumerate}
    \item Choose a tuple of permutations of $\Z_s$, $\pi=(\pi_1,\dots,\pi_r)$ u.a.r.
    \item Choose a bijection $\mu:[r] \to [r]$ u.a.r.
    \item Choose a map $\sigma:[r] \to [k]$ according to the distribution described above~\cref{lem:spread}.
    \item For $y=(y_1,\dots,y_k) \in \Z_s^k$, define $x_{\pi\sigma\mu}(y) = \paren{\pi_1(y_{\sigma(\mu^{-1}(1))}),\dots,\pi_r(y_{\sigma(\mu^{-1}(r))})}$.
    \item Accept iff $f'(y):=f(x_{\pi\sigma\mu}(y))$ is junta-degree-$d$.
    \end{enumerate}
\end{mdframed}

\noindent {\bf Remark.} We note that the $\pi_i$'s in the second version of the test do not correspond to the same $\pi_i$'s as in the first version of the test, but rather their compositions. However, we see that each time a variable $x_{\mu(i)}$ is substituted in the first version of the test, we also apply a fresh random permutation. In particular, this means that the ultimate result is that $x_{\mu(i)}$ is substituted by $\pi_i(y_{\sigma(\mu(i))})$, where $\pi_1,\ldots, \pi_r$ are independent random permutations of $\Z_s.$ Hence, the above test is indeed equivalent.

Let $\delta=\delta_d(f_r)= \delta(f,P) \le \varepsilon_1$ where $P:\Z_s^{[r]} \to \cG$ is junta-degree-$d$ and $E \subseteq \s^r$ be the points where $f$ and $P$ differ. Our objective is to show that
\begin{align}
\label{eqn:ed}
    \Pr_{\pi,\sigma,\mu}[T_r\text{ rejects }f] = \Pr_{\pi,\sigma,\mu}[f'\text{ is not junta-degree-}d] \ge \varepsilon\delta.
\end{align}

\noindent Let the functions $f':\s^k \to \cG$ and $P':\s^k \to \cG$ be defined by $f'(y)=f(x_{\pi\sigma\mu}(y))$ and $P'(y)=P(x_{\pi\sigma\mu}(y))$ respectively (these functions depend on $\pi,\sigma,\mu$) and $E' \subseteq \s^k$ be the points where these two restricted functions differ. 

To proceed further, we will need a subset $U$ of $\Z_s^k$ with the following properties (we defer the proof of this claim to~\cref{sec:noble}):

\begin{claim}
\label{clm:noble}
    Let $w = \ceil{\log(8\psi s^2) d} < k$. There exists a set $U \subseteq \Z_s^k$ of size $2^{w}$ such that
    \begin{enumerate}
        \item (Code) \label{noble:code} For all $y \ne y' \in U$, $$k/4 \le \Delta(y,y') \le 3k/4$$ where $\Delta(y,y')$ denotes the number of coordinates where $y$ and $y'$ differ.
        \item (Hitting) \label{noble:hit} No two junta-degree-$d$ functions $P:\s^k \to\cG$ and $Q:\s^k \to \cG$ can differ at exactly one point in $U$.
    \end{enumerate}
\end{claim}

Let $V= \{x_{\pi\sigma\mu}(y): y\in U\} \subseteq \Z_s^r$. 
Because the mapping $y \mapsto x_{\pi\sigma\mu}(y)$ is one-one conditioned on $\sigma$ being good, it holds that $\abs{V \cap E} = \abs{U \cap E'}$ under this conditioning. Now suppose the randomness is such that $\abs{U \cap E'} = 1$. Then, since no two junta-degree-$d$ functions can disagree at exactly one point in $U$ (Property~\ref{noble:hit} of~\cref{clm:noble}), it must be the case that $f'$ be of junta-degree greater than $d$ (as $P'$, being a restriction of a junta-degree-$d$ function is already junta-degree-$d$). Therefore, for~\eqref{eqn:ed} we can set $\varepsilon:=\Pr_{\sigma}[\sigma~\text{~good}] \ge 1/2^{O(k)}$ and show
    \begin{align*}
         \Pr_{\substack{\pi,\sigma,\mu}} [\abs{U \cap E'}=1 \mid \sigma\text{~good}] = \Pr_{\substack{\pi,\mu\\ \sigma~\text{~good}}}[\abs{U \cap E'}=1] \ge \delta.\end{align*}
    
    \noindent By a simple inclusion-exclusion, the above probability is 
    \begin{align}
        \Pr_{\substack{\pi,\mu\\ \sigma\text{~good}}}[\abs{V\cap E}=1] & \ge \sum_{y \in U} \Pr_{\substack{\pi,\mu\\ \sigma\text{~good}}}[ x_{\pi\sigma\mu}(y) \in E] - \sum_{y \ne y' \in U} \Pr_{\substack{\pi,\mu\\ \sigma\text{~good}}}[ x_{\pi\sigma\mu}(y) \in E \text{ and } x_{\pi\sigma\mu}(y') \in E]
        \label{eqn:ie}
    \end{align}
    For any $y\in U$, $x_{\pi\sigma\mu}(y) = \paren{\pi_1(y_{\sigma(\mu^{-1}(1))}),\dots,\pi_r(y_{\sigma(\mu^{-1}(r))})}$ is uniformly distributed over $\Z_s^r$ since $\pi_1,\dots,\pi_r$ are random permutations of $\Z_s$. Hence the first part of~\eqref{eqn:ie} is 
    \begin{align}
        \sum_{y \in U} \Pr[ x_{\pi\sigma\mu}(y) \in E] = |U|\cdot \frac{\abs{E}}{s^r} = \abs{U}\cdot \delta.\label{eqn:first-part}
    \end{align}
    
    \noindent For any fixed $y \ne y' \in U$ and good $\sigma$, we claim that the random variables $x:=x_{\pi\sigma\mu}(y)$ and $x':=x_{\pi\sigma\mu}(y')$ are related as follows:
    
    \begin{claim}
$x' \sim S_\nu(x)$, for some $\nu \in [1/32,31/32]$.
\label{clm:delta}
    \end{claim}

    \begin{proof}

        Let $A = \set{i \in [k]:y_i = y'_i}$ and
        $$J = \{j\in [r]:y_{\sigma(\mu^{-1}(j))}=y'_{\sigma(\mu^{-1}(j))}\}=\set{j\in [r]:x_j = x'_j}.$$ Note that $J$ only depends on $\mu$ (and not $\pi$). Moreover since $\mu:[r]\to [r]$ is a uniformly random bijection, $J$ is a uniformly random subset of size $\sum_{j\in A} \abs{\sigma^{-1}(j)}$. For any $j\in [r]$,
        \begin{align*}
                & \Pr_{\pi,J}[x_j = a\text{~and~}x'_j=b \mid J \ni j] =\begin{cases}
                \case 1,\text{if $a=b$~and~}\\
                \case 0\text{~otherwise.}
                \end{cases}
        \end{align*}
        
         \begin{align*}
                 & \Pr_{\pi,J}[x_j = a\text{~and~}x'_j=b \mid J \not\owns j] \\
                &~~~~~~~~= \begin{cases}
                \case 0,\text{if $a=b$~and~}\\
                \case \underset{\substack{\pi,J,\\z:=y_{\sigma(\mu(j))},z':=y'_{\sigma(\mu(j))}}}{\Pr}[\pi_j(z) = a\text{~and~}\pi_j(z')=b \mid z\ne z']=\frac{1}{s(s-1)}\text{~otherwise.}
                \end{cases}
        \end{align*}
        For the last equality above, we are using the fact that a random permutation of $\Z_s$ sends two distinct elements to two different locations uniformly at random. Further, 
        the joint variables $(x_j,x'_j)$ are mutually independent across $j$ as the permutations $\{\pi_j\}_{j}$ are mutually independent. Hence, $x' \sim S_\nu(x)$ where $\nu = 1-\abs{J}/r = \sum_{j\in [k]\setminus A} \abs{\sigma^{-1}(j)}/r \ge \paren{k-\abs{A}}/4k \ge \frac{k/4}{4k} \ge 1/16$, where the first inequality because $\sigma$ is good and the second one is using Property 1 that points in $U$ satisfy. For an upper bound on $\nu$, we note that $\abs{J} = \sum_{j\in  A} \abs{\sigma^{-1}(j)} \ge \frac{r}{4k}\cdot \abs{A} \ge r/16$. 
    \end{proof}

    \noindent 
    
    \noindent For the second term of~\eqref{eqn:ie},
    \begin{align}
        \Pr_{\substack{\pi,\mu \\ \sigma\text{~good}}}[ x_{\pi\sigma\mu}(y) \in E \text{ and } x_{\pi\sigma\mu}(y') \in E] & = \Pr_{\substack{x \sim \Z_s^r\\ x' \sim S_\nu(x)}}[x \in E\text{ and }x' \in E] \tag{for some $\nu \in [1/32,31/32]$ depending on $\sigma$, using~\cref{clm:delta}}\nonumber\\
        & \le C\cdot\delta^{1+\lambda}\text{~~~for some constant $C$ and $\lambda=1/2^{14}\log s$.}
        \label{eqn:second-part} \intertext{\hfill (Using spherical noise small-set expansion (\cref{thm:small-set-3}))} \nonumber
    \end{align}
    Plugging the bounds~\eqref{eqn:first-part} and~\eqref{eqn:second-part} back in~\eqref{eqn:ie}, we get
    \begin{align*}
        \Pr_{\substack{\pi,\mu,\\ \sigma~\text{good}}}[\abs{V \cap E} = 1] & \ge \abs{U}\delta - \abs{U}^2C\delta^{1+\lambda} \ge \abs{U}\delta/2 \ge \delta.
    \end{align*}
    The above inequalities follow from $\abs{U} = 2^{w}$ and $\delta \le \varepsilon_1$; this is where we set 
    $$\varepsilon_1 := \paren{1/2C\abs{U}}^{1/\lambda} = \paren{1/2C2^w}^{2^{14}\log s}\ge 1/s^{O(\log(8\psi s^2) d)} \ge 1/s^{O(k)}.$$
    Hence we conclude that $$\Pr_{\pi,\sigma,\mu}[T_r\text{ rejects }f] \ge \Pr_{\sigma}[\sigma~\text{good}]\cdot \Pr_{\substack{\pi,\mu\\ \sigma~\text{good}}}[\abs{U \cap E'}=1] \ge \varepsilon \Pr_{\substack{\pi,\mu\\ \sigma~\text{good}}}[\abs{V \cap E} = 1] \ge \varepsilon\delta.$$
\end{proof}

\subsubsection{Existence of $U$}
\label{sec:noble}



In order to prove~\cref{clm:noble} we will need the following:

\begin{claim}
\label{clm:noble-0}
        There exists a matrix $M \in \F_2^{k \times w}$ such that $U := \{Mz : z \in \F_2^{w}\} \subseteq \F_2^k$ is if size $2^{w}$ and:
        \begin{itemize}
            \item For all $ y\ne y' \in U $, we have $k/4 \le \Delta(y,y') \le 3k/4$.
            \item There exists a function $\chi: U \to \{\pm 1\}$ such that for all $I \in \binom{[k]}{\le d}$, we have 
            $$\sum_{y \in U:~y^I=1^I} \chi(y) = 0.$$
        \end{itemize}
   \end{claim}
   \begin{proof}
        We will show that picking $M$ uniformly at random satisfies both the items with positive probability. For Item 1, if suffices if for all $y \ne 0^k$ in $U$, $k/4 \le \# y \le 3k/4$; that is, for all $z\in \F_2^{a} \setminus \{0^{d+1}\}$, $k/4 \le \# Mz \le 3k/4$. For any fixed $z\ne 0^{a}$, we note that $y=Mz$ is uniformly distributed over $\F_2^k$ as $M\sim \F_2^{k\times w}$. Hence, by a Chernoff bound, we have $\Pr_{M} \brac{\# Mz \notin \brac{k/4,3k/4}} \le 2e^{-k/24}$. A union bound over all $z\in \F_2^{w} \setminus \{0^{w}\}$ gives

        $$\Pr_M \brac{\neg\paren{\forall y\ne y' \in U,~k/4\le \Delta(y,y') \le 3k/4}} \le 2^{w} \cdot 2e^{-k/24} <1/2.$$ However, it is a known fact that a uniformly chosen rectangular matrix has full rank with probability at least $1/2$. Therefore, with positive probability there must be a matrix $M$ such that it is full rank and Item 1 holds. We fix such an $M$ and prove Item 2.\\
        
        \noindent For $y\in U$, as $M$ is full rank there exists a unique $z\in \F_2^w$ such that $Mz = y$. Then we define
        $$\chi(y):=(-1)^{\innerprod{z,\eta}} = (-1)^{z_1 \eta_1 + \dots + z_{w} \eta_{w}},$$
        where $\eta \in \F^w_2$ is an arbitrary vector such that it is not in the $\F_2^w$-span of any $d$ rows of $M$. Such an $\eta$ always exists as the number of vectors that {\em can} be expressed as a linear combination of $d$ rows of $M$ is at most 
        \begin{align*}
        \binom{k}{d} 2^d \le \paren{\frac{ek}{d}}^d 2^d = \paren{2e\psi s^2}^d < 2^{w},\end{align*}
        the total number of vectors in $\F^w_2$.\\
        
        \noindent Let $M_1,\dots,M_k \in \F^w_2$ denote the rows of $M$. For any $I \in \binom{[k]}{\le d}$ and $y=Mz$, the condition $y^I=1^I$ is equivalent to: $\innerprod{z,M_i} = 1$ for all $i\in I$. Hence we have 
            \begin{align}
            \label{eqn:sum}
            \sum_{y \in U:~y^I=1^I} \chi(y) & = \sum_{z \in \F_2^a:~\forall i\in I,~\innerprod{z,M_i}=1} \paren{-1}^{\innerprod{z,\eta}}
            \end{align}
        As $\eta$ is linearly independent with $\{M_i\}_{i\in I}$, there exists $\eta' \ne 0^a$ such that $\innerprod{\eta',M_i} = 0$ for all $i\in I$ and $\innerprod{\eta',\eta} = 1$: this is because we can treat these conditions as a system of linear equations over $\F_2$.
        
        \noindent Note that for any $z \in \F^w_2$, $\innerprod{z,M_i} = 1$ if and only if $\innerprod{z+\eta', M_i}= \innerprod{z, M_i} +\innerprod{\eta', M_i} =1$. Since $z\ne z+\eta'$, we may partition the summation~\eqref{eqn:sum} into buckets of size 2, each bucket corresponding to $z$ and $z+\eta'$ for some $z$. For each such bucket, the sum is $$(-1)^{\innerprod{z,\eta}} + (-1)^{\innerprod{z+\eta',\eta}} = (-1)^{\innerprod{z,\eta}} + (-1)^{\innerprod{z,\eta} + \innerprod{\eta',\eta}} =  (-1)^{\innerprod{z,\eta}}\paren{1+(-1)^{\innerprod{\eta' , \eta}}} =0,$$ so the overall sum is also 0.
   \end{proof}

Using this result, we will prove~\cref{clm:noble}:

\begin{proof}[Proof of~\cref{clm:noble}]
   Identifying the $0$'s (resp.~$1$'s) in $\F_2$ and $\Z_s$, we will rephrase the above claim as $U$ being a subset of $\Z_s^k$ instead of $\F^k_2$: Then, this set $U \subseteq \{0,1\}^k \subseteq \Z_s^k$ immediately satisfies Item 1 of~\cref{clm:noble}. For Item 2, it suffices to show that for any non-zero junta-degree-$d$ function $P:\Z_s^k \to \cG$,
   \begin{align}
   \label{eqn:dotpdt}
        \sum_{y \in U} P(y) \cdot \chi(y) = 0
   \end{align} where $\chi:\{0,1\}^k \to \Z$ is from~\cref{clm:noble-0}. 
   Towards a contradiction, suppose that a junta-degree-$d$ function $P$ has exactly one point in $y^* \in U$ such that $P(y^*) \ne 0$. Then using~\eqref{eqn:dotpdt}, $ 0 + \dots +0 + P(y^*) \cdot \chi(y^*) + 0 + \dots + 0 = 0$ and as $\chi(y^*)=\pm 1$, $P(y^*)=0$, a contradiction. \\
   
\noindent    To prove~\eqref{eqn:dotpdt}, we expand $P$ into its junta-polynomial representation:
   \begin{align*}
        \sum_{y \in U} P(y) \cdot \chi(y) & = \sum_{y\in U} \sum_{\substack{a \in {\Z_s}^k \\ \#a \le d}} g_{a} \cdot \paren{\prod_{i\in [k]:~a_i \ne 0}\delta_{a_i}(y_i)}\chi(y) \\
        & = \sum_{\substack{a \in {\Z_s}^k \\ \#a \le d}} g_{a} \cdot \paren{\sum_{y\in U} \chi(y) \prod_{i\in [k]:~a_i \ne 0}\delta_{a_i}(y_i)}.
   \end{align*}

 \noindent   For any $a\in \Z_s^k$, letting $I:=\{i\in [k]:a_i \ne 0\}$, the inner factor is
   \begin{align*}
        \sum_{y\in U} \chi(y) \prod_{i\in [k]:~a_i \ne 0}\delta_{a_i}(y_i) = \sum_{y\in U:~y^I = a^I} \chi(y).
   \end{align*}
   Now if $a$ contains any coordinates taking values other than $0$ and $1$, the above sum is $0$ since all the coordinates of $y \in U$ are either $0$ or $1$. On the other hand, if $a \in \{0,1\}^k$, then $a^I = 1^I$ and~\cref{clm:noble-0} is applicable, again giving a sum of $0$. Therefore, 
   \begin{align*}
        \sum_{y\in U} P(y) \cdot \chi(y) = \sum_{\substack{a \in {\s}^k \\ \#a \le d}} g_{a} \cdot \paren{\sum_{y\in U} \chi(y) \prod_{i\in [k]:~a_i \ne 0}\delta_{a_i}(y_i)} = \sum_{\substack{a \in {\s}^k \\ \#a \le d}} g_{a} \cdot 0  = 0.
   \end{align*}
   
\end{proof}

\subsection{Large-distance lemma}
\label{sec:large-distance}

\begin{proof}[Proof of~\cref{lem:large}]

    For this proof, we may assume without loss of generality that $I=[r]$ as relabelling the variables does not affect the probability of a random restriction (i.e., $x_j=\pi_j(x_i)$) being $\varepsilon_0$-close to junta-degree-$d$.
    We will prove the contrapositive: assuming $\delta_d(f_{r-1}) \le \varepsilon_0$ for more than $k^2/2r(r-1)$ fraction of choices of $(i,j,\pi_j)$ (call these {\em bad} restrictions), we will construct a junta-degree-$d$ function $P$ such that $\delta(f_r, P) \le \varepsilon_1$.
    Like in~\cite{BKSSZ,bafna2017local}, the high level idea is to ``stitch'' together low-junta-degree functions corresponding to the restrictions $f_{r-1}$ (which we shall call $P^{(h)}$) into a low-junta-degree function that is close to $f_r$.

    \subsubsection{Two cases}
    \label{sec:two-cases}

    As there are more than $\frac{k^2}{2r(r-1)} r(r-1)s! = k^2 s!/2$ many bad tuples $(i,j,\pi_j)$, by pigeon-hole principle, there must be some permutation $\pi:\Z_s \to \Z_s$ such that the number of bad tuples 
    of the form $(i,j,\pi)$ is more than $k^2/2$. In fact, we can say something more:
    
    Consider the directed graph $G_{\text{bad}}$ over vertices $[r]$ with a directed edge $(i,j)$ for each bad tuple  $(i,j,\pi)$. As the number of edges in $G_{\text{bad}}$ is at least $k^2/2$, by the pigeon-hole principle, we can conclude that there is a {\em matching} or a {\em star}\footnote{A matching is a set of disjoint edges and a star is a set of edges that share a common start vertex, or a common end vertex.} in $G_{\text{bad}}$ of size $L:=k/4$. For the rest of the proof, we will handle both the cases in parallel as the differences are minor.

    Suppose we are in the matching case and the corresponding bad tuples are 
    \begin{align*}
        (i_1,j_1,\pi), \dots (i_L,j_L,\pi),
    \end{align*}
    where $i_1,\dots,i_L,j_1,\dots,j_L$ are all distinct. Let $\gamma$ denote the identity permutation of $\Z_s$. Consider the function $\widetilde{f}_r(x_1,\dots,x_r)$ obtained by replacing the variables $x_{i_1},\dots,x_{i_L}$ in  $f_r$ with $\pi(x_{i_1}),\dots,\pi(x_{i_L})$ respectively. Then, $\delta_d(\widetilde{f}_r) = \delta_d(f_r)$ and $(i_h,j_h,\pi)$ is a bad restriction for $f_r$ if and only if $(i_h,j_h,\gamma)$ is a bad restriction for $\widetilde{f}_r$, for all $h\in [L]$. Moreover, if $\widetilde{f}_r$ satisfies  $\delta(\widetilde{P},\widetilde{f}_r) \le \varepsilon_1$, then there also exists a junta-degree-$d$ $P$ such that  $\delta(P,f_r) \le \varepsilon_1$ (obtained from $\widetilde{P}$ by applying the inverse permutation $\pi^{-1}$ to $x_{i_1},\ldots, x_{i_L}$). Therefore, without loss of generality we may assume that $\pi = \gamma$ to construct a junta-degree-$d$ function  $P$ such that $\delta(P,f_r) \le \varepsilon_1$. A similar reduction holds in the star case. \\
    
\noindent     We may further assume w.l.o.g.~that the matching case corresponds to the tuples $$(L+1,1,\gamma),(L+2,2,\gamma),\dots(2L,L,\gamma)$$ and the star case corresponds to $$(r,1,\gamma),(r,2,\gamma),\dots,(r,L,\gamma).$$
    
 \noindent   For $h\in [L]$, we define 
    \begin{align*}
        R_h := \begin{cases}
        \{x \in \Z_s^r : x_{L+h} = x_{h}\} \text{ in the matching case,}\\
        \{x \in \Z_s^r : x_h = x_r\} \text{ in the star case.}
    \end{cases}
    \end{align*}
    as the points that agree with the $h$-th bad restriction $(i,j,\pi)$ in the matching or star case correspondingly. Let $R_h'$  denote the complement of $R_h$. Then for any function $P:\Z_s^r \to \cG$,
    \begin{align}
    \label{eqn:prob-error}
        \Pr_{x\sim \Z_s^r} \brac{f_r(x) \ne P(x)} & \le \Pr_x\brac{x \notin \bigcup_{h\le L} R_h} \\ & ~~~~ + \sum_{h\le L}  \Pr_x \brac{x\in R_{h} \setminus \bigcup_{h'<h} R_h'}\cdot \Pr_x\brac{f_r(x)\ne P(x) \left| x\in R_h\setminus \bigcup_{h'<h} R_{h'}\right.}
    \end{align}

    \noindent To estimate the above probabilities, we note that in both the matching or the star case, 
        \begin{align}\label{eqn:prob-error-1}
        \Pr_{x \sim \Z_s^r} \brac{x\notin \bigcup_{h \le L} R_h} = \Pr_x \brac{x \in \bigcap_{h \le L} R_h'} = \paren{1-\frac{1}{s}}^L,\end{align}

        \noindent and 
        \begin{align}\label{prob-error-2}
            \Pr_{x \sim \Z_s^{r}}\brac{x \in R_h \setminus \bigcup_{h'<h} R_{h'}}  = \Pr_x \brac{x \in R_h \cap \bigcap_{h'<h} R_{h'}'} = \frac{1}{s} \paren{1-\frac{1}{s}}^{h-1}.
        \end{align}

   \noindent For $h\in [L]$, let $f_{r-1}^{(h)}:\Z_s^r \to \cG$ be the restricted function corresponding to the $h$-th bad tuple, treated as a function of all the $r$ many variables (rather than $r-1$). Let $P^{(h)}$ denote the junta-degree-$d$ function that is of distance at most $\varepsilon_0$ from $f_{r-1}^{(h)}$. We will show that there is a junta-degree-$d$ function $P$ that agrees with $P^{(h)}$ over $R_h$, for all $h$.

    \begin{claim}
        \label{clm:low-deg-interpolation}
        There exists a junta-degree-$d$ function $P$ such that $P(x) = P^{(h)}(x)$ for all $h \in [L]$ and $x \in R_h$.
    \end{claim}

    \noindent Assuming the above claim, for such a $P$ and any $h\le L$, we have
    \begin{align*}
        \Pr_x \brac{f_r(x) \ne P(x) \left| x\in R_h \setminus \bigcup_{h'<h} R_{h'}\right.} & \le \frac{\Pr_x \brac{f_r(x) \ne P(x) \mid x\in R_h}}{\Pr\brac{x \in R_h \setminus \bigcap_{h'<h} R_{h'} \mid x\in R_h}} \\
        & = \frac{\Pr_x \brac{f_r(x) \ne P(x) \mid x\in R_h}}{\paren{1-\frac{1}{s}}^{h-1}} \\
        & = \paren{\frac{s}{s-1}}^{h-1} \Pr_{x \sim R_h} \brac{f_{r-1}^{(h)}(x) \ne P^{(h)}(x)}\\
        & \le  \paren{\frac{s}{s-1}}^{h-1} \varepsilon_0.
    \end{align*}

\noindent  Then we can bound~\eqref{eqn:prob-error} as $\Pr_{x\sim \Z_s^r} \brac{f_r(x) \ne P(x)} \le \paren{1-\frac{1}{s}}^L + \frac{L\varepsilon_0}{s} \le \varepsilon_1/2 + \varepsilon_1/2 = \varepsilon_1$ as we can set $\varepsilon_0 :=2s\varepsilon_1/k \ge 1/s^{O(k)}$ and 
    \begin{align*}\paren{1-\frac{1}{s}}^{k/4} \le e^{-k/4s} = e^{-\psi sd/4} \le \frac{1}{2}\paren{\frac{1}{2C2^{\lceil{\log(8\psi s^2) d}\rceil}}}^{2^{14}\log s} = \frac{\varepsilon_1}{2}.\tag{for the last inequality, we can take $\psi$ to be a sufficiently large constant}\end{align*}
    
\end{proof}

\subsubsection{Low-junta-degree interpolation}
\label{sec:low-junta-interpolate}

We now give a proof of claim~\cref{clm:low-deg-interpolation}.

\begin{proof}[Proof of~\cref{clm:low-deg-interpolation}]

For each $h\le L$, we view $P^{(h)}:\Z_s^r \to \cG$ as a function with variables $x_1,\dots,x_r$. For $h' \ne h$, we define the function $P^{(h)}|_{h'}:\Z_s^r \to \cG$ to be the restricted function of $P^{h}$ obtained by identifying the variables corresponding to the restriction $R_{h'}$; note that we still view $P^{(h)}|_{h'}$ as a function over all the $r$ variables. 

 \noindent    To construct $P$, we will use the fact that the $P^{(h)}$ are already very similar to each other in the following sense: 
    For any $h\ne h' \le L$, $P^{(h)}|_{h'} = P^{(h')}|_h$ as functions. To see this, we note that
    \begin{align*}
        \Pr_{x\sim \Z_s^r} \brac{P^{(h)}|_{h'} (x) \ne P^{(h')}|_h (x)} & = \Pr_{x \sim R_{h}\cap R_{h'}} \brac{P^{(h)}|_{h'} (x) \ne P^{(h')}|_h (x)} \tag{as both the functions do not depend on the variables removed corresponding to the restrictions}\\
        & = \Pr_{x \sim R_{h}\cap R_{h'}} \brac{P^{(h)} (x) \ne P^{(h')} (x)}\\
        & \le \Pr_{x\sim R_{h}\cap R_{h'}} \brac{P^{(h)} (x) \ne f_r(x)} +  \Pr_{x\sim R_{h}\cap R_{h'}} \brac{P^{(h')}(x) \ne f_r(x)}\\
        & = \Pr_{x\sim R_{h}\cap R_{h'}} \brac{P^{(h)} (x) \ne f_{r-1}^{(h)}(x)} +  \Pr_{x\sim R_{h}\cap R_{h'}} \brac{P^{(h')}(x) \ne f_{r-1}^{(h')}(x)}\\
        & \le \frac{\Pr_{x\sim R_{h}} \brac{P^{(h)} (x) \ne f_{r-1}^{(h)}(x)}}{\Pr_{x\sim R_h} \brac{x \in  R_{h'}}} +  \frac{\Pr_{x\sim R_{h'}} \brac{P^{(h')}(x) \ne f_{r-1}^{(h')}(x)}}{\Pr_{x\sim R_{h'}} \brac{x \in  R_{h}}} \\
        & \le s\varepsilon_0 + s\varepsilon_0 < 1/s^d. \tag{as $\varepsilon_0 \le \varepsilon_1$ can be made sufficiently small by choosing $\psi$ large enough}
    \end{align*}

\noindent     By~\cref{clm:junta-polys-sz}, since two different junta-degree-$d$ functions must differ on at least $1/s^d$ fraction of inputs, we conclude that 
    \begin{align}
    \label{eqn:similar}
    P^{(h)}|_{h'} = P^{(h')}|_h.
    \end{align} 

\noindent Since, any function has a unique junta-polynomial representation (\cref{clm:junta-unique}) we may see that $P^{(h)}|_{h'} = P^{(h')}|_h$ even as junta-polynomials. Now, we shall construct a junta-polynomial of degree at most $d$ such that $P|_h = P^{(h)}$, where $P|_h$ denotes the junta-polynomial obtained by identifying the variables corresponding to the restriction $R_h$. \\

\noindent Suppose we are in the star case.
Our objective is to find group elements $(g_a)_{\substack{a\in \s^r\\\# a \le d}}$ such that $$P(x)=\sum_{\substack{a\in \s^r\\\# a \le d}} g_a \cdot \prod_{i\in [r]:a_i \ne 0} \delta_{a_i} (x_i)$$ and $P|_h = P^{(h)}$ for all $h\le L$.
However instead of the above formulation, it will be cleaner to set up our unknowns as follows: Find group elements $(g_a)_{\substack{a\in \s^r\\\# a \le d}}$ such that 
\begin{align}
\label{eqn:diff-junta-poly}
P(x)=\sum_{\substack{a\in \Z_s^r\\\# a \le d}} g_a \cdot \prod_{i\in [L]:a_i \ne 0} \paren{\delta_{a_i} (x_i) - \delta_{a_i}(x_r)} \prod_{i\in [L+1..r]:a_i \ne 0} \delta_{a_i} (x_i)\end{align} 
and $P|_h = P^{(h)}$ for all $h\le L$. Surely, if such $(g_a)_a$ exist then upon expanding the products involving $\delta_{a_i} (x_i) - \delta_{a_i}(x_r)$ in \eqref{eqn:diff-junta-poly}, we will get a junta-polynomial of degree at most $d$. Hence, it suffices if we ensure that $P|_h = P^{(h)}$ and solve~\eqref{eqn:diff-junta-poly}. 

Let us make the representation in~\eqref{eqn:diff-junta-poly} more formal with the following claim, that is proved at the end of the section.
\begin{claim}
    \label{clm:diff-junta} Any junta-degree-$d$ function $Q:\Z_s^r \to \cG$ can be uniquely expressed as
    \begin{align*}
       Q(x_1,\dots,x_r)=\sum_{\substack{a\in \Z_s^r\\\# a \le d}} g_a \cdot \prod_{i\in [L]:a_i \ne 0} \paren{\delta_{a_i} (x_i) - \delta_{a_i}(x_r)} \prod_{i\in [L+1..r]:a_i \ne 0} \delta_{a_i} (x_i),
    \end{align*}
    where $g_a\in \cG$.
\end{claim}

\noindent For any $h \le L$, since $P|_h$ corresponds to the restriction $x_h = x_r$, projecting onto the summands without the factor $\delta_{a_i}(x_h) - \delta_{a_i}(x_r)$ computes $P|_h$. That is,
$$P|_h(x)=\sum_{\substack{a\in \s^r\\\# a \le d\\ a_h=0}} g_a \cdot \prod_{i\in [L]:a_i \ne 0} \paren{\delta_{a_i} (x_i) - \delta_{a_i}(x_r)} \prod_{i\in [L+1..r]:a_i \ne 0} \delta_{a_i} (x_i).$$


\noindent Now for any $h\le L$, since $P^{(h)}$ is a junta-degree-$d$ function, by~\cref{clm:diff-junta} it can be expressed as:
\begin{align}
\label{eqn:p(h)}
    P^{(h)}(x) = \sum_{\substack{a \in \Z_s^r\\ \#a \le d\\ a_h=0}} g^{(h)}_a \cdot \prod_{i \in [L]:a_i \ne 0} \paren{\delta_{a_i}(x_i) - \delta_{a_i}(x_r)} \prod_{i \in [L+1..r]:a_i \ne 0} \delta_{a_i}(x_i)
\end{align}
for some $g^{(h)}_a \in \cG$. We are only summing over $a_h=0$ since $P^{(h)}(x)$ does not depend on the variable $x_h$.
From \eqref{eqn:p(h)}, setting $x_{h'} = x_r$ for any $h' \ne h \in [L]$, 
\begin{align}
\label{eqn:p(h)h'}
    P^{(h)}|_{h'}(x) = \sum_{\substack{a \in \s^r\\ \#a \le d\\ a_h=0\\ a_{h'}=0}} g^{(h)}_a \cdot \prod_{i \in [L]:a_i \ne 0} \paren{\delta_{a_i}(x_i) - \delta_{a_i}(x_r)} \prod_{i \in [L+1..r]:a_i \ne 0} \delta_{a_i}(x_i)
\end{align} and 
\begin{align}
\label{eqn:p(h')h}
    P^{(h')}|_{h}(x) = \sum_{\substack{a \in \s^r\\ \#a \le d\\ a_h=0\\ a_{h'}=0}} g^{(h')}_a \cdot \prod_{i \in [L]:a_i \ne 0} \paren{\delta_{a_i}(x_i) - \delta_{a_i}(x_r)} \prod_{i \in [L+1..r]:a_i \ne 0} \delta_{a_i}(x_i)
\end{align}

Since~\eqref{eqn:similar} says that $P^{(h)}|_{h'} = P^{(h')}|_h$ as functions, they must have the same junta-polynomial representations. In fact, by the uniqueness property of~\cref{clm:diff-junta} for $Q=P^{(h)}|_{h'} - P^{(h')}|_h$, we notice that the expressions in~\eqref{eqn:p(h)h'} and~\eqref{eqn:p(h')h} must be identical. Hence comparing coefficients, for all $a\in \s^n$ such that $\# a \le d$ and $a_h = a_{h'}=0$, we have that
\begin{align}
\label{eqn:consistent}
    g^{(h)}_a = g^{(h')}_a.
\end{align}

Now we are ready to describe the solution $(g_a)_a$ to~\eqref{eqn:diff-junta-poly} and thus the function $P$: For $a\in \s^r$ such that $\#a \le d$,
we set $g_a = g^{(h)}_a$ if $a_h=0$ and $0$ otherwise. It is important to note that the choice of $h$ here does not matter due to~\eqref{eqn:consistent}. To verify that $P|_h = P^{(h)}$ for all $h\le L$, we observe that 
\begin{align*}
    P|_h(x) & = \sum_{\substack{a\in \s^r\\\# a \le d\\ a_h=0}} g_a \cdot \prod_{i\in [L]:a_i \ne 0} \paren{\delta_{a_i} (x_i) - \delta_{a_i}(x_r)} \prod_{i\in [L+1..r]:a_i \ne 0} \delta_{a_i} (x_i)\\
        & = \sum_{\substack{a\in \s^r\\\# a \le d\\ a_h=0}} g^{(h)}_a \cdot \prod_{i\in [L]:a_i \ne 0} \paren{\delta_{a_i} (x_i) - \delta_{a_i}(x_r)} \prod_{i\in [L+1..r]:a_i \ne 0} \delta_{a_i} (x_i) = P^{(h)}(x).
\end{align*}

We now prove~\cref{clm:diff-junta}.

\begin{proof}[Proof of~\cref{clm:diff-junta}]
    By~\cref{clm:junta-unique}, we know that $Q$ can be expressed as a junta-polynomial of degree at most $d$ where the monomials are of the form $\prod_{i\in [r]:~a_i \ne 0} \delta_{a_i}(x_i)$ for some $a\in \Z_s^r$ such that $\#a \le d$. We can express this as 
    \begin{align*}
    \prod_{i\in [r]:~a_i \ne 0} \delta_{a_i}(x_i) & = \prod_{i\in [L]:a_i\ne 0} \delta_{a_i}(x_i) \prod_{i\in [L+1..r]:a_i \ne 0}\delta_{a_i}(x_i) \\
    & = \prod_{i\in [L]:a_i\ne 0} \brac{(\delta_{a_i}(x_i)-\delta_{a_i}(x_r))+\delta_{a_i}(x_r)} \prod_{i\in [L+1..r]:a_i \ne 0}\delta_{a_i}(x_i)
    \end{align*}
    Expanding the above expression in terms of $(\delta_{a_i}(x_i)-\delta_{a_i}(x_r))$ by using the rule
    $$\delta_{b}(x_r) \delta_c(x_r) = \begin{cases}\case \delta_{b}(x_r)\text{~if $b=c$}\\ \case 0\text{~otherwise}\end{cases},$$ 
    we get the desired representation. To prove uniqueness, we will appeal to the uniqueness of the junta-polynomial representation (i.e.,~\cref{clm:junta-unique}). Thus, it suffices to show that for any
    \begin{align}
    \label{eqn:23}
        Q(x_1,\dots,x_r) = \sum_{\substack{a\in \Z_s^r \\ \#a \le d}} g_a \cdot \prod_{i\in [L]:a_i\ne 0} (\delta_{a_i}(x_i) - \delta_{a_i}(x_r)) \prod_{i\in [L+1..r]:a_i\ne 0} \delta_{a_i}(x_i)
    \end{align} 
    such that $g_{a}\ne 0$ for some $a$, there would be at least one non-zero coefficient in its junta-polynomial representation, which can be obtained by simplifying~\eqref{eqn:23} into a junta-polynomial. We have two cases:
    \begin{itemize}
        \item \noindent {\bf Case 1: $g_a = 0$ for all $a\in \Z_s^r$ such that $\exists i\in [L]:a_i \ne 0$.} In this case, the RHS of~\eqref{eqn:23} is directly a junta-polynomial with at least one non-zero coefficient.

        \item \noindent {\bf Case 2: $g_a \ne 0$ for some $a\in \Z_s^r$ such that $\exists i\in [L]:a_i \ne 0$.} Fix such an $a=a^*$. The only summand in~\eqref{eqn:23} that produces the monomial $\prod_{i\in [r]: a^{*}_i\neq 0} \delta_{a^{*}_i}(x_i)$ is the one corresponding to $a=a^*$. Hence, the coefficient of $\prod_{i\in [r]: a^{*}_i\neq 0} \delta_{a^{*}_i}(x_i)$ in the junta-polynomial representation of $Q$ is $g_{a^*} \ne 0$.
    \end{itemize}
\end{proof}

Continuing the proof of~\cref{clm:low-deg-interpolation}: The matching case is similarly handled: Our objective now is to find $(g_a)_{\substack{a\in \Z_s^r\\ \#a\le d}}$ such that 
\begin{align}
    \label{eqn:p-match}
    P(x) = \sum_{\substack{a\in \Z_s^r \\ \# a \le d}} g_a \cdot \prod_{i\in [L]:a_i \ne 0} (\delta_{a_i} (x_i)-\delta_{a_i}(x_{L+i})) \prod_{i\in [L+1..r]:a_i \ne 0} \delta_{a_i} (x_i)
\end{align} and $P|_h = P^{(h)}$ for all $h\le L$. Similar to~\cref{clm:diff-junta}, it can be proven that the above representation is unique for any given junta-degree-$d$ function $P$. For any $P$ of the form~\eqref{eqn:p-match}, since $P|_h$ corresponds to the restriction $x_h = x_{L+h}$ we have
\begin{align}
\label{eqn:ph-match}
    P|_h(x) = \sum_{\substack{a\in \Z_s^r\\ \#a \le d\\ a_h = 0}} g_a \cdot \prod_{i\in [L]:a_i \ne 0} (\delta_{a_i}(x_i) - \delta_{a_i}(x_{L+i})) \prod_{i\in [L+1..r]:a_i \ne 0} \delta_{a_i}(x_i)
\end{align}
    Since $P^{(h)}$ is junta-degree-$d$, we can express it as 
    \begin{align*}
            P^{(h)}(x) = \sum_{\substack{a\in \Z_s^r\\ \#a \le d\\ a_h = 0}} g_a^{(h)} \cdot \prod_{i\in [L]:a_i \ne 0} (\delta_{a_i}(x_i) - \delta_{a_i}(x_{L+i})) \prod_{i\in [L+1..r]:a_i \ne 0} \delta_{a_i}(x_i),
    \end{align*} for some $g^{(h)}_a \in \cG$. Therefore,
    \begin{align}
    \label{eqn:p(h)h'-match}
    P^{(h)}|_{h'}(x) = \sum_{\substack{a\in \Z_s^r\\ \#a \le d\\ a_h = 0\\ a_{h'}=0}} g_{a}^{(h)} \cdot \prod_{i\in [L]:a_i\ne 0} (\delta_{a_i}(x_i) - \delta_{a_i}(x_{L+i})) \prod_{i\in [L+1..r]:a_i \ne 0} \delta_{a_i}(x_i)
    \end{align} and 
    \begin{align}
    \label{eqn:p(h')h-match}
     P^{(h')}|_{h}(x) = \sum_{\substack{a\in \Z_s^r\\ \#a \le d\\ a_h = 0\\ a_{h'}=0}} g_{a}^{(h')} \cdot \prod_{i\in [L]:a_i\ne 0} (\delta_{a_i}(x_i) - \delta_{a_i}(x_{L+i})) \prod_{i\in [L+1..r]:a_i \ne 0} \delta_{a_i}(x_i)
    \end{align}
    Again, since~\eqref{eqn:similar} says that $P^{(h)}|_{h'} = P^{(h')}|_h$ as functions, the representations on the RHS of~\eqref{eqn:p(h)h'-match} and~\eqref{eqn:p(h')h-match} must be identical; i.e., for all $a\in \Z_s^r$ such that $\# a \le d$ and $a_h = a_{h'} = 0$, we have
    \begin{align}
    \label{eqn:consistent-match}
    g_a^{(h)} = g_{a}^{(h')}.
    \end{align}
    Now we will construct $P$ by setting $g_a=g_a^{(h)}$ if $a_h=0$ and $0$ otherwise in~\eqref{eqn:p-match}. The choice of $h$ here does not matter because of~\eqref{eqn:consistent-match}. For this $P$, we can see that
    \begin{align*}
        P|_h(x) & = \sum_{\substack{a\in \Z_s^r\\ \#a \le d\\ a_h =0}} g_a\cdot \prod_{i\in[L]:a_i \ne 0} (\delta_{a_i}(x_i) - \delta_{a_i}(x_{L+i})) \prod_{i\in [L+1..r]:a_i\ne 0} \delta_{a_i}(x_i) \\
        & = \sum_{\substack{a\in \Z_s^r\\ \#a \le d\\ a_h =0}} g^{(h)}_a \cdot \prod_{i\in [L]:a_i\ne 0}(\delta_{a_i}(x_i) - \delta_{a_i}(x_{L+i})) \prod_{i\in [L+1..r]:a_i\ne 0}\delta_{a_i}(a_i) = P^{(h)}(x).
    \end{align*}
    This finishes the proof of~\cref{clm:low-deg-interpolation} and the large-distance lemma.
\end{proof}

\section{Low-degree testing}
\label{sec:low-degree-test}
Moving on to the second half of this paper, we will describe our low-degree test now.

\paragraph{The degree test~(\degtest):}

\label{sec:overview:degree}
Given query access to $f:\cS^n \to \F$, the following test (called~\degtest) works to test whether $f$ is degree-$d$. We may assume that $s=\abs{\cS} \ge 2$ as $f$ is a constant function otherwise. 
\begin{mdframed}
    {\bf Test~$\degtest(f)$: gets query access to $f:\cS^n \to \F$}
    \begin{itemize}
        \item Run~$\juntadegtest(f)$ to check if $f$ is junta-degree-$d$.
        \item Run~$\weakdegtest(f)$.
        \item Accept iff both the above tests accept.
    \end{itemize}
\end{mdframed}

In the above description, the sub-routine~\weakdegtest~corresponds to the following test.

\begin{mdframed}
    {\bf Test $\weakdegtest(f)$: gets query access to $f:\cS^n \to \F$}
    \begin{itemize}
        \item Choose a map $\mu:[n] \to [K]$ u.a.r.~where $K=t(d+1)$ and $t=s^3$.
        \item For $y=(y_1,\dots,y_{K}) \in \cS^{K}$, define $x_\mu(y) = (y_{\mu(1)},\dots,y_{\mu(n)})$.
        \item Define the function $f':\cB(\cS,t)^{K/t}\to \F$ as $f'(y) = f(x_\mu(y))$, where $\cB(\cS,t)$, defined in~\cref{sec:prelims}, is the ``balanced'' subset of $\cS^t$.
        \item Accept iff $f'$ is degree-$d$.
    \end{itemize}
\end{mdframed}



 We now move on to the analysis of this test, proving ~\cref{thm:deg-test}. We prove a slightly different version below which would imply~\cref{thm:deg-test} by simply repeating the $\degtest$ test $O(1/\varepsilon)$ times.


\begin{theorem}
    For any subset $\cS \subseteq \F$ of size $s$ if a function $f:\cS^n \to \F$ is degree-$d$, then~\degtest~always accepts. Otherwise,
    $$ \Pr[\degtest\text{~rejects~}f] \ge \varepsilon\delta$$  for some $\varepsilon=(sd)^{-O(s^3 d)}$, where $\delta=\delta'_d(f)$ is the distance to the degree-$d$ family.
\end{theorem}

\begin{proof}
    Let $g:\cS^n \to \F$ be a closest junta-degree-$d$ function to $f$ i.e., $\delta_d(f) = \delta(f,g)$.
    There are four cases, three of which are trivial. Here $\varepsilon_2 = {K}^{-O(K)}$ is a small enough threshold to be decided later.
    \begin{itemize}
        \item {\bf Case 1: $\delta'_d(f)=0$ i.e., $f$ is degree-$d$.} 
        
        If $f$ has a degree-$d$ polynomial, then~\juntadegtest~always succeeds. As identifying variables does not decrease the degree,~\weakdegtest~also succeeds. Hence, the~\degtest~test always accepts.
         
        \item {\bf Case 2: $\delta_d(f) > \varepsilon_2$.} 

        In this case, the~\juntadegtest~part of~\degtest~rejects $f$ with probability at least $s^{-O(s^2 d)} \delta_d(f) \ge K^{-O(K)}$. Hence, $$\Pr[\degtest\text{~rejects~}f] \ge (sd)^{-O(s^3d)} \ge (sd)^{-O(s^3d)}\delta'_d(f).$$
        
        \item {\bf Case 3: $\delta_d(f) \le \varepsilon_2$ and $g$ is degree-$d$.} Since any degree-$d$ function is also junta-degree-$d$ and $g$ is the closest junta-degree-$d$ function to $f$, we have $\delta(f,g) = \delta_d(f) = \delta'_d(f)$. Therefore,
        $$\Pr[\degtest\text{~rejects~}f] \ge \Pr[\juntadegtest\text{~rejects~}f] \ge \varepsilon \delta_d(f) = \varepsilon \delta'_d(f).$$
        
        \item {\bf Case 4: $\delta_d(f) \le \varepsilon_2$ but $g$ is not degree-$d$.} 

        Let $E\subseteq \cS^n$ be the points where $f$ and $g$ differ; thus $\abs{E}/s^n = \delta_d(f) \le \varepsilon_2$. Let $$V_\mu = \set{x_\mu(y): y\in \cB(\cS,t)^{K/t}}.$$
        Suppose $\mu:[n]\to [K]$ is such that $V_\mu \cap E = \emptyset$ and $\weakdegtest~\text{~rejects~}g$. Then ~\weakdegtest~does not distinguish between $f$ and $g$ and hence rejects $f$ as well. We will show that both these events occur with good probability. For the first probability, we will upper bound
        \begin{align*}
            \Pr_\mu \brac{V_\mu \cap E \ne \emptyset} & = \Pr_\mu\brac{\exists y\in\cB(\cS,t)^{K/t} : x_\mu(y) \in E} \\
            & \le \abs{\cB(\cS,t)^{K/t}} \cdot \Pr_\mu \brac{\text{For fixed arbitrary $y\in \cB(\cS, t)^{K/t}$}, x_\mu(y)\in E}.
        \end{align*}
        Note that since all points in $\cB(\cS, t)^{K/t}$ contain an equal number of occurrences of all the elements of $\cS$, $x_\mu(y)$ is uniformly distributed in $\cS^n$ for a uniformly random $\mu$. Hence, the above probability is
        \begin{align*}
            \Pr_\mu \brac{V_\mu \cap E \ne \emptyset} & \le \abs{\cS}^{K} \cdot \frac{\abs{E}}{s^n} \le s^{K}\varepsilon_2 < \frac{1}{2K^{d}}.\tag{by setting $\varepsilon_2:=1/4s^KK^{d} \ge K^{-O(K)}$}
        \end{align*}
        The remainder of this section is dedicated to proving the following claim.
        \begin{claim}
        \label{clm:weakdegtestg}
            $\Pr_\mu [\weakdegtest~\text{rejects}~g] \ge 1/K^{d}.$
        \end{claim}
        
        Assuming this claim, 
        \begin{align*}\Pr[\degtest~\text{rejects}~f] & \ge \Pr[\weakdegtest~\text{rejects}~f]\\
         & \ge \Pr[\weakdegtest~\text{rejects}~g] - \Pr[V_\mu \cap E \ne \emptyset] \\
         & \ge \frac{1}{K^{d}} - \frac{1}{2K^{d}} = \frac{1}{2K^{d}} \ge \frac{\delta'_d(f)}{2K^d}.\end{align*}
    \end{itemize}
    This finishes the analysis of the low-degree test assuming~\cref{clm:weakdegtestg}.
\end{proof}

\subsection{Soundness of~\weakdegtest}

\begin{proof}[Proof of~\cref{clm:weakdegtestg}]
We will need the following lemma about the vector space formed by functions over $\cB(\cS,t)^{K/t} \subseteq \cS^K$.

\begin{lemma}
\label{lem:basis}
    For $\cT\subseteq \cS^K$, the vector space of functions from $\cT$ to $\F$ has a basis $\set{m_1,\dots,m_\ell}$ such that for any $f:\cT\to \F$ of the form $f=c_1 m_1 + \dots c_\ell m_\ell$ for some $c_i \in \F$, we have
    \begin{align}
        f\text{~is degree-}d \iff \forall i,~c_i = 0\text{~or~}m_i\text{~is degree-}d.
    \end{align}
\end{lemma}

\begin{proof}
    ($\impliedby$) direction is easy to prove as a sum of degree-$d$ functions can be computed by summing the degree-$d$ polynomials of the components. To prove ($\implies$) direction of the claim, we will construct the basis $B=\{m_1,\dots,m_\ell\}$ as follows:
    \begin{itemize}
        \item Initialize $B= \emptyset$.
        \item For $i=0$ to $K$, 
            \begin{itemize}
                \item while there is some degree-$i$ function $m$ that is not in the $\F$-linear span of $B$, add $m$ to $B$.
            \end{itemize}
    \end{itemize}
    If $f$ is degree-$d$, the above procedure ensures that either $f\in B$ or $f$ is in the span of some degree-$d$ basis functions in $B$. In either case, $f$ is expressible as a linear combination of degree-$d$ basis functions in $B$. 
\end{proof}

\noindent Let $g':\cB(\cS,t)^{K/t} \to \F$ be defined as $g'(y) = g(x_\mu(y))$ where $x_\mu(y)=(y_{\mu(1)},\dots,y_{\mu(n)})$. Then, recall that $\weakdegtest$ rejects $g$ iff $g'$ is not degree-$d$. We use~\cref{lem:basis} above to fix a suitable basis $m_1,\ldots, m_\ell$ for functions from $\cT = \cB(\cS,t)^{K/t}$ to $\F.$ Then we can write each $g'(y)$ obtained above uniquely as
\[
g'(y) = \sum_{i=1}^\ell c_i \cdot m_i
\]
where the coefficients $c_i$ are some functions of $\mu$. We will treat $\mu:[n]\rightarrow [K]$ as a random element of $[K]^n$.\\

\noindent We argue that each $c_i$ is junta-degree-$d$ (as a function of $\mu$). To see this, we recall that $g$ is junta-degree-$d$. Consider the case when $g$ is a function of only $x_{i_1},\ldots,x_{i_s}$ for some $s\leq d$. In this case, clearly the polynomial $g'$ depends only on $\mu_{i_1},\ldots, \mu_{i_s}$. In particular, each $c_i$ is just an $s$-junta. Extending the argument by linearity, we see that for any $g$ that is junta-degree-$d$, the underlying coefficients $c_i$ of $g'(y)$ are junta-degree-$d$ polynomials in the co-ordinates of $\mu.$\\

\noindent Now assume that there exists a $\mu^*:[n] \to [K]$ such that $g'(y)$ is not degree-$d$ (we will show the existence of such a ``good'' $\mu$ in the next subsection). Thus by~\cref{lem:basis} there exists $i^*\in [\ell]$ such that $m_{i^*}$ is not degree-$d$ and $c_{i^*}(\mu^*) \ne 0$. In particular, the function $c_{i^*}$ is non-zero.\\ 

\noindent We have argued that there is an $m_{i^*}$ in the basis such that the associated coefficient $c_{i^*}$ is a non-zero junta-degree-$d$ polynomial. In particular,~\cref{clm:junta-polys-sz} implies that the probability that $c_i^{*}(\mu) \neq 0$ for a random $\mu$ is at least $1/K^d.$ Thus,
    \begin{align*}
        \Pr_\mu\brac{\weakdegtest\text{~rejects~}g} & = \Pr_\mu\brac{g'\text{~is not degree-}d}\\
        & \ge \Pr_\mu \brac{c_{i^*}(\mu) \ne 0}\tag{using~\cref{lem:basis}}\\
        &\geq 1/K^{d}.
    \end{align*}
\end{proof}

\subsection{Existence of a good~$\mu$}

We will show for any function $g:\cS^n \to \F$ that is not degree-$d$, there exists a map $\mu:[n] \to [K]$ such that the function $g'(y) = g(x_\mu(y))$ defined for $y\in \cB(\cS,t)^{K/t}$ is also not degree-$d$. This is easy to prove if the domain of $g'$ were to be $\cS^K$, but is particularly tricky in our setting. 

Let $D=d+1$. We will give a map $\mu:[n] \to [t]\times [D] \equiv [tD] = [K]$ instead. Let $P$ be the polynomial with individual degree at most $s-1$ representing $g$; suppose the degree of $P$ is $d'>d$ and let $m(x)=c\cdot x_{i_1}^{a_{1}}\cdots x_{i_\ell}^{a_{\ell}}$ be a monomial of $P(x)$ of degree $d'$ for some non-zero $c\in \F$, where $a_{j} \ge 1$ for all $j$ and $i_1,\dots,i_\ell$ are some distinct elements of $[n]$ and $\ell \le d$ as $g$ is junta-degree-$d$. Then, we define $\mu$ as follows for $i\in [n]$:
\begin{align*}
    \mu(i) = \begin{cases}
    \case (1,j),\text{~if}~i=i_j\text{ for some }j\in [\ell]\\
    \case (1,D),\text{~otherwise.}
    \end{cases}
\end{align*}
It is easy to inspect that $P(x_\mu(y))$ (call it $Q(y)$) is a polynomial in variables $y_{(1,1)},\dots,y_{(1,D)}$, and is of degree $d'>d$ -- this is because the monomial $m(x)$, upon this substitution turns to 
$$m(x_\mu(y)) = c \cdot y_{(1,1)}^{a_1} \cdots y_{(1,\ell)}^{a_\ell},$$ which cannot be cancelled by $m'(x_\mu(y))$ for any other monomial $m'(x)$ of $P(x)$, as if $m'$ contains the variable $x_i$ for some $i\notin \{i_1,\dots,i_{\ell}\}$ then $m'(x_\mu(y))$ contains the variable $y_{(1,D)}$ and on the other hand if $m'$ only contains variables $x_i$ for some $i\in \{i_1,\dots,i_{\ell}\}$, then the individual degree of $y_{(1,j)}$ in the two substitutions differs for some $j$. Hence, the degree of $Q(y)$ is $a_1+\cdots + a_\ell =d'$. As we can express the function $y_{(1,D)}^{a}$ for $a>s-1$ as a polynomial in $y_{1,D}$ of individual degree at most $s-1$, we can further transform $Q(y)$ so that it has individual degree at most $s-1$, while maintaining the properties that it still only contains the variables $y_{(1,1)},\dots,y_{(1,D)}$ (i.e., the first ``row'') and has degree $d'$ and computes the function $g'(y)$. The following claim then completes the proof of the existence of a good $\mu$ by setting $w=D$ and $d'=d'$. 

\begin{claim}
    \label{clm:goodmu}
    For formal variables $y\equiv (y_1,\dots,y_w) \equiv (y_{(i,j)})_{(i,j) \in [t] \times [w]}$, let $Q(y)$ be a polynomial of degree $d'\ge 0$ containing only the variables from the first row. Then the degree of $Q(y)$ as a function over $\cB(\cS,t)^{w}$ is exactly $d'$.
\end{claim}
\begin{proof}
The proof is by induction on $w$. The base case $w=1$ is crucial and it is equivalent to the following claim:
\begin{claim}
    \label{clm:goodmufinal}
    For $0\le d'\le s-1$, the function $f_{d'}:\cB(\cS,t)\to \F$ defined as $f_{d'}(z) = z_1^{d'}$ for $z=(z_1,\dots,z_t) \in \cB(\cS,t)$ has degree exactly $d'$.
\end{claim}

Assuming the above claim, let $w>1$ be arbitrary. As $d'=0$ is trivial to handle, we will assume that $d' \ge 1$. Hence, $Q$ contains at least one monomial $m$ of degree $d'$ and containing some variable $y_{(1,j)}$ with individual degree $a \in [s-1]$. Without loss of generality, suppose $j=w$. Since~\cref{clm:goodmufinal} states that the function $y_{(1,w)}^a$ is linearly independent of degree-$(a-1)$ functions over $\cB(\cS,t)$, there exists a function $C:\cB(\cS,t) \to \F$ such that for any $f:\cB(\cS,t) \to \F$ 
\begin{align}
\label{eqn:sum-dot}
    \innerprod{C, f} = \begin{cases}
    \case 1\text{~if~}f=y_{(1,w)}^a\text{~i.e.,~}a\text{-th power of the last coordinate}\\
    \case 0\text{~if~}f\text{~is degree-}(a-1).
    \end{cases}
\end{align}
Now we decompose $Q$ as a polynomial over variables in the first $w-1$ columns and coefficients being the monomials over variables in the last column: that is
\begin{align}
\label{eqn:q-decomp}
    Q(y) = \sum_{\alpha \in [0..a]} Q'_\alpha(y_1,\dots,y_{w-1}) \cdot y_{(1,w)}^{\alpha},
\end{align}
where $y_j$ represents the variables in the $j$-th column $Q'_{a}\ne 0$ has degree $d'-a$. Here, we are using the fact that $Q$ only contains variables from the first row. \\

\noindent Towards a contradiction, suppose there is some degree-$(d'-1)$ polynomial $R(y)$ such that $Q(y) = R(y)$ for all $y\in \cB(\cS,t)^w$. We may decompose $R$ as follows:
\begin{align*}
    R(y) = \sum_{\alpha \in [0..s-1]^t} R'_\alpha(y_1,\dots,y_{w-1}) \cdot y_w^{\alpha} 
\end{align*} where $y_w^{\alpha} = y_{(1,w)}^{\alpha_1} \cdots y_{(t,w)}^{\alpha_t}$ and for all $\alpha$, either $R'_\alpha=0$ or is of degree (as a formal polynomial) at most $d'-1-\abs{\alpha}_1$, where $\abs{\alpha}_1 = \alpha_1 + \dots + \alpha_t$. \\

\noindent Fixing $y_1,\dots,y_{w-1} \in \cB(\cS,t)$ to arbitrary values and treating $Q(y)$ and $R(y)$ as functions of $y_w$, we get
\begin{align*}
    \innerprod{C, Q(y_1,\dots,y_{w-1},y_w)} & = \innerprod{C, \sum_{\alpha\in [0..a]} Q'_\alpha(y_1,\dots,y_{w-1})\cdot y^\alpha_{(1,w)}} \\
    & = \sum_{\alpha\in [0..a]} Q'_\alpha(y_1,\dots,y_{w-1})\cdot \innerprod{C,y^\alpha_{(1,w)}}\\
    & = Q'_a(y_1,\dots,y_{w-1}). \tag{using~\eqref{eqn:sum-dot}}
\end{align*} Similarly,
\begin{align*}
    \innerprod{C, R(y_1,\dots,y_{w-1},y_w)} & = \innerprod{C, \sum_{\alpha\in [0..s-1]^t} R'_\alpha(y_1,\dots,y_{w-1})\cdot y^\alpha_{w}} \\
    & = \sum_{\alpha\in [0..s-1]^t} R'_\alpha(y_1,\dots,y_{w-1})\cdot \innerprod{C,y_w^\alpha}\\
    & = \sum_{\alpha\in [0..s-1]^t:~\abs{\alpha}_1 \ge a} R'_\alpha(y_1,\dots,y_{w-1})\cdot \innerprod{C,y_w^\alpha} \tag{using~\eqref{eqn:sum-dot}}
\end{align*} As a polynomial in the variables of $y_1,\dots,y_{w-1}$, the final expression above is of degree at most $d'-1-\abs{\alpha}_1 \le d'-1-a$. However, as a function it is equivalent to $Q'_a$, which has a strictly higher degree, $d'-a$. This contradicts the induction hypothesis.
\end{proof}

We end with a proof of~\cref{clm:goodmufinal}.
\begin{proof}[Proof of~\cref{clm:goodmufinal}]
    The upper bound of $d'$ on the degree is trivial. For the lower bound, it suffices to prove the claim for $d'=s-1$ as we can obtain a polynomial computing $z_1^{s-1}$ by multiplying $z_1^{s-1-d'}$ to any polynomial computing $z_1^{d'}$. 

    The proof is by contradiction. Suppose $z_1^{s-1}$ is degree-$(s-2)$ as a function. By symmetry, then $z_i^{s-1}$ is also degree-$(s-2)$ for all $i\in [t]$. Recall that any function $f:\cB(\cS,t)\to \F$ can be expressed as a (not necessarily unique) polynomial $P$ of individual-degree at most $s-1$. Replacing the factor $z_i^{s-1}$ in each monomial of $P$ with corresponding degree-$(d-2)$ functions, we observe that $f$ can be expressed as a polynomial of individual-degree at most $s-2$. The monomials of individual-degree at most $s-2$ form a spanning set for the functions computed by such polynomials. Its dimension is $(s-1)^{t}$. On the other hand, the vector space of all functions $f:\cB(\cS,t)\to \F$ is of dimension $\abs{\cB(\cS,t)} = \frac{t!}{(\paren{t/s}!)^s}$, which can be shown to be strictly larger than $(s-1)^t$ by using $t=s^3$ and Stirling's estimates. Thus, there exists a function $f:\cB(\cS,t)\to \F$ which cannot be represented as a polynomial of individual-degree at most $s-2$, a contradiction.  
\end{proof}

\section{More general domains}
\label{sec:gen-dom}

In the previous sections, we gave a junta-degree-$d$ test for functions of the form $f:\cS^n \to \cG$ and a degree-$d$ test for functions of the form $f:\cS^n \to \F$ for $\cS \subseteq \F$. More generally, one can ask the question of testing (junta-)degree-$d$ functions over an arbitrary grid $\cS_1 \times \dots \times \cS_n$. We show that this is possible for junta-degree but not for degree testing in general.

\subsection{Junta-degree testing}
\label{sec:rectangular}

We will now prove~\cref{thm:gen-dom-junta} for general grids $\cS_1 \times \dots \times \cS_n$ using the fact that it was already proven to be correct if one has $\abs{\cS_i}=s\ge 2$ for all $i$, in~\cref{sec:low-junta-degree-test}.

\begin{proof}[Proof of~\cref{thm:gen-dom-junta} for general grids]
    Let $s_i=\abs{\cS_i}$, $s=\max_i s_i$ and $L=\lcm(s_1,\dots,s_n) \le s!$ be the least common multiple of the set sizes. Note that $L\ge 2$, unless all the set sizes are equal to $1$ in which case the junta-degree test is trivial. Without loss of generality, we can assume that $\cS_i = [s_i]$. Let 
    $f:[s_1] \times \dots \times [s_n]\to \cG$ be the function we want to test. Let $\cT_i = [s_i] \times [L_i]$ where $L_i=L/s_i$. Then we will transform $f$ to a function $f_\lambda$ over the domain $\cT_1 \times \dots \times \cT_n$ so that the junta-degree-$d$-ness is preserved. Then, we can apply the junta-degree test for $f_\lambda$ which is now a grid in which $\abs{\cT_i}=L$ is constant across all $i$. This is stated in more detail below:
    \begin{mdframed}
    \label{alg:asym}
    \noindent {\bf Test $\juntadegtest(f)$: gets query access to $f:[s_1] \times \dots [s_n] \to \cG$}
    \begin{itemize}
    \item Define $f_\lambda:\cT_1 \times \dots \times \cT_n \to \cG$ as 
    $$f_\lambda((x_1,y_1),\dots,(x_n,y_n)) = f(x_1,\dots,x_n).$$
    \item Accept $f$ iff $f_\lambda$ is junta-degree-$d$ (check this using $\juntadegtest$ from~\cref{sec:low-junta-degree-test}; note $\abs{\cT_i}$ is identical for all $i\in [n]$).
    \end{itemize}
\end{mdframed}

    If is clear that $f$ is junta-degree-$d$ iff $f_\lambda$ is junta-degree-$d$. We need to argue that the distance is preserved as well, i.e., $\delta_d(f) = \delta_d(f_\lambda)$. Then, the probability that the above test rejects $f$ is equal to the probability of $\juntadegtest(f_\lambda)$ rejecting which is proportional to $\delta_d(f_\lambda)=\delta_d(f)$. \\
    
    \noindent Let $g$ be a closest junta-degree-$d$ function to $f$. Then $\delta_d(f_\lambda) \le \delta_d(f_\lambda,g_\lambda) = \delta_d(f,g) =\delta_d(f)$, where we used $\delta_d(f_\lambda,g_\lambda) = \delta_d(f,g)$ since $f_\lambda$ and $g_\lambda$ can be treated as functions over $[s_1] \times [L_1] \times \dots \times [s_n] \times [L_n]$ but the output never depends on the values of the even coordinates. For the other direction, let $h:\cT_1 \times \dots \times \cT_n \to \cG$ be a closest junta-degree-$d$ function to $f_\lambda$. Let $h^b:[s_1] \times \dots \times [s_n] \to \cG$ be defined as 
    $h^b(x) = h((x_1,b_1),\dots,(x_n,b_n))$. We claim that there exists some $b$ such that $\delta(f_\lambda,h^b_\lambda) \le \delta(f_\lambda,h)$. This is because the expectation of $\delta(f_\lambda,h^b_\lambda)$ is $\delta(f_\lambda,h)$:
   \begin{align*}\E_{b\in [L_1] \times \dots \times [L_n]}\brac{\delta(f_\lambda,h^b_\lambda)}
   & = \E_{b\in [L_1] \times \dots \times [L_n]}\brac{\delta(f,h^b)}\\
   & = \E_{b\in [L_1] \times \dots \times [L_n]}\brac{\Pr_{x\in [s_1] \times \dots \times [s_n]} \brac{f(x)\ne h^b(x)}} \\
   & = \Pr_{(x,b)}\brac{f(x) \ne h(x,b)}
    = \delta(f_\lambda,h).\end{align*}
    
 \noindent    Hence for such a $b$,
    \begin{align*}
    \delta_d(f) & \le \delta(f,h^b)\tag{as $h^b$ is junta-degree-$d$}\\
    & = \delta(f_\lambda,h^b_\lambda)\\
    & \le \delta(f_\lambda, h) = \delta_d(f_\lambda).
    \end{align*}

\end{proof}

\subsection{Degree testing impossibility}
\label{sec:asymmetric}
In the previous section, we have shown that low-degree functions over $\cS^n$ are locally testable. It is natural to ask whether this can be extended to general product domains of the form $\cS_1 \times \dots \times \cS_n$ (for potentially different sets $\cS_i$). The previous subsection addresses this question for junta-degree. Unlike junta-degree testing, we will prove that this is not possible in general for degree testing. In fact, the lower bound (on the number of queries) we prove is in a scenario where $\abs{\cS_i} = 3$ and degree $d=1$.

\begin{proof}[Proof of~\cref{thm:impossible}]
    Let $\F$ be any field of size at least $n+2$ with distinct elements $\{0,1,a_1,\dots,a_n\}$. For $i\in [n]$, let $\cS_i = \{0,1,a_i\}$. We will refer to $0,1$ as Boolean elements and the remaining as non-Boolean ones. Let $\zeta(b)=b$ if $b$ is Boolean and $\star$ otherwise.

    As degree-$d$ functions over $\cS_1 \times \dots \times \cS_n$ form a linear subspace over $\F$, by~\cite{Ben-SassonHR05} any test can be converted to a one-sided, non-adaptive one without changing the number of queries or the error by more than a factor of $2$. Thus, without loss of generality let~\Test~be a one-sided, non-adaptive test for the degree-$1$ family of the following form, where $\cD$ is some distribution over matrices $\F^{\ell \times n}$ with the row vectors from $\cS_1 \times \dots \times \cS_n$ and $P:\F^\ell \to \{\text{true},\text{false}\}$ is some predicate.
    \begin{mdframed}
        \begin{itemize} 
            \item Sample $M \sim \cD$ and let the rows of $M$ be $x^{(1)},\dots,x^{(\ell)} \in \cS_1 \times \dots \times \cS_n$.
            \item Query $f$ to construct the vector $f_M = (f(x^{(1)}),\dots,f(x^{(\ell)})) \in \F^\ell$.
            \item Accept $f$ iff $P(f_M)$ is true.
    \end{itemize}
    \end{mdframed} 
We will show that if~\Test~accepts degree-$1$ functions with probability 1 and rejects $\Omega(1)$-far functions with probability $\Omega(1)$, then $\ell=\Omega(\log n)$. \\
    
\noindent     As~\Test~is one-sided, for all $M\in \supp(\cD)$, 
    \begin{align}
    \label{eqn:gM}
        f_M \in \colspace(M) \implies P(f_M)\text{~is true}.
    \end{align} 
    
    \noindent Here {\em $\supp$}($\cdot$) refers to the support and {\em $\colspace$}($\cdot$) is the column space. Fix an arbitrary $M\in \supp(\cD)$. For an $i\in [n]$, suppose that there exists a $j\ne i\in [n]$ such that 
    \begin{align} 
    \label{eqn:star}
        \forall k\in [\ell],~\zeta(x^{(k)}_i) = \zeta(x^{(k)}_j). 
    \end{align} That is, the $i$-th and $j$-th columns of $M$ are ``identical'', upon the identification of $a_i$ and $a_j$ (with the symbol $\star$).
    Let $g(x)=x_i(x_i-1)$ be a function with domain $\cS_1 \times \dots \cS_n$ and co-domain $\F$; note that $g$ is 0 iff the $i$-th coordinate is Boolean. We note that $g$ is not degree-$1$, indeed it is $1/3$-far from being degree-$1$ (by an application of~\cref{clm:junta-polys-sz}). Therefore, the $k$-th coordinate of $g_M$ is $0$ if $x^{(k)}_i$ is Boolean and $a_i(a_i-1) \ne 0$ otherwise. Consider the column vector $M_i - M_j$, where $M_i$ denotes the $i$-th column of $M$. By~\eqref{eqn:star}, its $k$-th coordinate is 
    \begin{align*}   
        (M_i - M_j)_k = \begin{cases}
            \case 0,\text{~if~}\zeta(x^{(k)}_i)\ne \star,\\
            \case a_i - a_j \ne 0,\text{~otherwise}.
        \end{cases}
    \end{align*}
    
    \noindent Therefore, $g_M = \frac{a_i(a_i-1)}{a_i-a_j} (M_i - M_j)$ is in the column space of $M$, so~\Test~accepts $g$ because of~\eqref{eqn:gM}. As $\zeta(M_i):=(\zeta(x^{(k)}_i))_{k\in [\ell]}$ can take at most $3^\ell$ distinct values, for any $M\in \supp(\cD)$ there are at most $3^\ell$ many $i\in [n]$ such that there does not exist a $j\ne i$ such that~\eqref{eqn:star}~holds. Call such $i\in[n]$, {\em bad for $M$}. Therefore,
    \begin{align*}
    \Pr_{i\in [n]} \brac{\Test\text{~rejects~}x_i(x_i-1)} & = \Pr_{\substack{M\sim\cD\\i\in [n]}} \brac{\Test\text{~rejects~}x_i(x_i-1)}\\
    & \le \Pr_{\substack{M\sim\cD\\i\in [n]}} \brac{i\text{~bad for~}M} \le 3^\ell/n . 
    \end{align*}
    But the LHS is $\Omega(1)$ as for any $i\in [n]$, $x_i(x_i-1)$ is $\Omega(1)$-far from degree-$1$. Hence, the number of queries~\Test~makes is $\ell \ge \Omega(\log n)=\omega_n(1)$.
\end{proof}
\section{Small-set expansion for spherical noise} 
\label{sec:hypercontractive}
In this section, we will prove a small-set expansion theorem for spherical noise, which we have used for~\eqref{eqn:second-part} in the proof of the small-distance lemma  of junta-degree testing.\\




\noindent Let $f:\Z_s^n \to \C$ be an arbitrary function.

\begin{definition}[{\bf Bernoulli noise operator}]
    For $\nu \in [0,1]$,
    $$N_{\nu}f(x)=\E_{y \sim N_\nu(x)} \brac{f(y)}=\E_{y \sim \cD_\nu^{\otimes n}} [f(x+y)] .$$
\end{definition}

\begin{definition}[{\bf Spherical noise operator}]
    For $\nu \in [0,1]$ such that $\nu n \in \Z$, 
    $$S_{\nu}f(x)= \E_{y \sim S_\nu(x)}[f(y)]=\E_{y \sim \cE_\nu} [f(x+y)] .$$
\end{definition}

\noindent For the rest of this section, let $\rho\in [0,1]$ and $\nu = (1-1/s)(1-\rho) \in [0,1]$: it is easy to check that if each coordinate of $x$ is retained with probability $\rho$ and randomized (uniformly over $\Z_s$) with probability $1-\rho$, the resulting string is distributed according to $N_\nu(x)$.\\



\noindent We will use the following lemma:

\begin{lemma}[{\bf Small-set expansion for bernoulli noise,~e.g.~\cite{o2021analysis}}]
\label{lem:fourier-bern}
    Let $A \subseteq \Z_s^n$ be such that $\Pr_{x \sim \Z_s^n} \brac{x \in A} =\delta$. Then, for $q\ge 2$ and $0 \le \rho \le \frac{1}{q-1}\paren{\frac{1}{s}}^{1-2/q}$ we have
    \begin{align}
    \label{eqn:prob-bern}
        \Pr_{\substack{x\sim \Z_s^n\\y \sim N_\nu(x)}} \brac{x \in A \text{~and~} y\in A} \le \delta^{2-2/q},
    \end{align}
    where $\nu=(1-1/s)(1-\rho)$.
\end{lemma}

\noindent If $s\ge 3$, we reduce the case of spherical noise to bernoulli noise and then use the above lemma.

\begin{theorem}[{\bf Small-set expansion for spherical noise}]
\label{thm:small-set-3}
    Let $s\ge 3$ and $A\subseteq \Z_s^n$ be such that $\Pr_{x\sim \Z_s^n} \brac{x\in A} = \delta$. Then, for any $\nu \in [1/32,1]$
    \begin{align}
     \label{eqn:prob-sphere}
        \Pr_{\substack{x\sim \Z_s^n \\ y\sim S_\nu(x)}} \brac{x\in A \text{~and~} y\in A} \le 2\cdot \delta^{1+\lambda},
    \end{align}
    where $\lambda = \frac{1}{2^{14}\log s}$.
\end{theorem}

\begin{remark}
    When the size of the domain $s$ is equal to $2$ and $\nu \in [1/32,31/32]$, the above statement still holds (with the factor $2$ replaced with some other constant factor $C$) as proved by~\cite{polyanskiy2019hypercontractivity} (or Corollary~2.8 in~\cite{bafna2017local}).
\end{remark}

\begin{proof}
    Let $f:\Z_s^n \to \C$ be the indicator function of $A$ and consider its Fourier representation as in~\cref{def:fourier}:
    \begin{align*}
        f(x) = \sum_{\alpha \in \Z_s^n} \widehat{f}(\alpha) \chi_\alpha(x).
    \end{align*}
Then the probability in~\eqref{eqn:prob-bern} is equal to
 \begin{align}
       \Pr_{\substack{x\sim \Z_s^n\\y \sim N_\nu(x)}} \brac{x \in A \text{~and~} y\in A} = \E_{\substack{x\sim \Z_s^n\\ y \sim \cD_\nu^{\otimes n}}} \brac{f(x)f(x+y)} =  \sum_{\alpha \in \Z_s^n} \abs{\widehat{f}(\alpha)}^2 \E_{y\sim \cD_\nu^{\otimes n}} \brac{\chi_\alpha(y)} = \sum_{\alpha \in \Z_s^n} \abs{\widehat{f}(\alpha)}^2 \rho^{\#\alpha},\label{eqn:prob-bern-fourier}\intertext{\hfill{(using~\cref{lem:ns})}}\nonumber
\end{align}
and similarly the probability in~\eqref{eqn:prob-sphere} is equal to
    \begin{align}
        \label{eqn:prob-sphere-fourier}
        \Pr_{\substack{x\sim \Z_s^n \\ y\sim S_\nu(x)}} \brac{x\in A \text{~and~} y\in A} = \sum_{\alpha \in \Z_s^n} \abs{\widehat{f}(\alpha)}^2 \E_{y\sim \cE_\nu} \brac{\chi_\alpha(y)}.
    \end{align} 

\noindent We will show that for any $\alpha \in \Z_s^n$, the quantity $\E_{y\sim \cE_\nu} \brac{\chi_\alpha(y)}$ above is upper bounded by $2\cdot \widetilde{\rho}^{\#\alpha}$ for some constant $\widetilde{\rho}$.
    \begin{align}
        \E_{y \sim \cE_\nu}\brac{\chi_{\alpha}(y)} & = \E_{y \sim \cE_{\nu}}\brac{\chi_{\alpha_1}(y_1)\cdots \chi_{\alpha_n}(y_n)}\nonumber\\& = \E_{\substack{I \sim \binom{[n]}{\nu n}\\ \mu \sim \Z_s \setminus \{0\}\\y \sim 0^{\overline{I}}\circ \mu^I}}\brac{ \prod_{i \in I} \chi_{\alpha_i}(y_i) \prod_{i \notin I} \chi_{\alpha_i}(y_i) }\tag{where $I$ and $\mu$ are independent}\\& = 
        \E_{I,\mu,y}\brac{\prod_{i\in I} \chi_{\alpha_i}(\mu_i)}\tag{where $I \sim \binom{[n]}{\nu n}, \mu \sim \Z_s \setminus \{0\},$ and $y \sim 0^{\overline{I}}\circ \mu^I$}\\
        & = \E_{I} \E_{\mu \sim (\Z_s\setminus \{0\})^I}\brac{\prod_{i\in I} \chi_{\alpha_i}(\mu_i)}\nonumber\\
        & = \E_I \brac{\prod_{i\in I} \E_{\mu_i \sim \Z_s \setminus \{0\}}[\chi_{\alpha_i}(\mu_i)] }.\label{eqn:muii}
    \end{align}

\noindent Now we note that the inner term 
\begin{align}
\label{eqn:mui}
\E_{\mu_i \sim \Z_s \setminus \{0\}} \brac{\chi_{\alpha_i}(\mu_i)} = \begin{cases}
\case 1, \text{~if~}\alpha_i=0,\text{~and}\\
\case \frac{1}{s-1}\paren{\sum_{\mu_i \in\Z_s \setminus \{0\}} \chi_{\alpha_i}(\mu_i)}=\frac{1}{{s-1}}\paren{s\E_{\mu_i \sim \Z_s} [\chi_{\alpha_i}(\mu_i)]-1}=\frac{-1}{s-1}\text{~otherwise.}
\end{cases}
\end{align}
Therefore, denoting the coordinates of $\alpha$ with non-zero entries by $J \subseteq [n]$, plugging~\eqref{eqn:mui} into \eqref{eqn:muii} gives

\begin{align*}
    \E_{y\sim \cE_\nu} [\chi_\alpha(y)] & = \E_{I \sim \binom{[n]}{\nu n}} \brac{\paren{\frac{-1}{s-1}}^{\abs{J\cap I}}} \\ & \le \E_{I \sim \binom{[n]}{\nu n}} \brac{\paren{\frac{1}{2}}^{\abs{J \cap I}}}\tag{as $s\ge 3$}\\
    & \le \Pr_{I \sim \binom{[n]}{\nu n}} \brac{\abs{J \cap I} < \nu k/2} \cdot 1 + \E_{I \sim \binom{[n]}{\nu n}} \brac{\paren{\frac{1}{2}}^{\abs{J \cap I}} \middle\vert\ \abs{J \cap I} \ge \nu k/2}.
\end{align*}    

\noindent Denoting $\abs{J}=\#\alpha$ by $k$, we observe that $\abs{J \cap I}$ is distributed according to the hypergeometric distribution of $k$ draws (without replacement) from a population of size $n$ and $\nu n$ many success states. Hence, by a tail bound~\cite{hoeffding1994probability}  
$$\Pr[\abs{J \cap I} < \nu k/2] \le e^{-\nu^2 k/2}.$$Using $\nu \ge 1/32$,
\begin{align*}
    \E_{y\sim \cE_\nu} [\chi_\alpha(y)] & \le \Pr_{I \sim \binom{[n]}{\nu n}} \brac{\abs{J \cap I} < \nu k/2} \cdot 1 + \E_{I \sim \binom{[n]}{\nu n}} \brac{\paren{\frac{1}{2}}^{\abs{J \cap I}} \middle\vert\ \abs{J \cap I} \ge \nu k/2}\\
    & \le 2^{-k/2^{11}} + 2^{-k/64} \le 2\cdot 2^{-k/2^{11}} = 2\cdot \widetilde{\rho}^k
\end{align*}
where $\widetilde{\rho} := 2^{-2^{-11}}$.\\

\noindent Plugging the above bound in~\eqref{eqn:prob-sphere-fourier}, we get 
\begin{align*}
    \Pr_{\substack{x\sim \Z_s^n \\ y\sim S_\nu(x)}} \brac{x\in A \text{~and~} y\in A} & = \sum_{\alpha \in \Z_s^n} \abs{\widehat{f}(\alpha)}^2 \E_{y\sim \cE_\nu} \brac{\chi_\alpha(y)} \\
    & \le 2\sum_{\alpha \in \Z_s^n} \abs{\widehat{f}(\alpha)}^2 \widetilde{\rho}^{\#\alpha}\\
    & =2 \Pr_{\substack{x\sim \Z_s^n\\ y\sim N_{\widetilde{\nu}}(x)} } {\brac{{x\in A\text{~and~}y\in A}}}\tag{using, where $\widetilde{\nu} = (1-1/s)(1-\widetilde{\rho})$}\\
    & \le 2 \delta^{2-2/q}.
\end{align*}
For the last step above, we are using~\eqref{eqn:prob-bern-fourier} and~\cref{lem:fourier-bern} with $q:=2+\frac{1}{2^{12} \log s}=:2+\varepsilon$; here we can verify that 
\begin{align*}
    \frac{1}{q-1}\paren{\frac{1}{s}}^{1-2/q} \ge \frac{1}{e^\varepsilon}\paren{\frac{1}{s}}^{1-2/q} 
    \ge \frac{1}{e^\varepsilon}\paren{\frac{1}{s}}^{\varepsilon/3} 
     \ge s^{-2\varepsilon}  = \widetilde{\rho}
\end{align*} as needed to invoke the bernoulli small-set expansion lemma.
Hence, the above probability is at most $2\delta^{1+\lambda}$ for $\lambda=\varepsilon/4$.
\end{proof}

\bibliographystyle{alpha}
\bibliography{references}

\appendix
\addtocontents{toc}{\protect\setcounter{tocdepth}{1}}
	
\end{document}

%% file: macros.tex
\usepackage{fullpage}
\usepackage{amssymb}
\usepackage{amsmath}
\usepackage{amsthm}
\usepackage{multirow}
\usepackage{graphicx}
\usepackage{mathrsfs}
\usepackage[utf8]{inputenc}
\usepackage[noadjust]{cite}
\usepackage{mdframed}
\usepackage{color}
\usepackage[pdfstartview=FitH,pdfpagemode=UseNone,colorlinks=true,citecolor=blue,linkcolor=blue,backref=page]{hyperref}

\usepackage{mathtools}
\usepackage{relsize}
\usepackage{amsfonts}
\usepackage[T1]{fontenc}
\usepackage{mathpazo}
\usepackage{bm}
\usepackage{changepage}
\usepackage{thmtools}
\usepackage{thm-restate}
\usepackage{subcaption}
\usepackage{enumitem}
\usepackage{algorithm}
\usepackage{algpseudocode}
\usepackage{circuitikz}

\usepackage{nicematrix}
\usepackage{arydshln}

\usepackage{tikz}
\usetikzlibrary{positioning}
\usetikzlibrary{calc}
\usetikzlibrary{decorations.shapes}
\usetikzlibrary{decorations.pathreplacing,angles,quotes}
\usepackage{todonotes}

\usepackage[capitalize,nameinlink]{cleveref}

\allowdisplaybreaks

\definecolor{blueviolet}{RGB}{60,50,200}
\definecolor{oliveg}{RGB}{40,200,30}
\hypersetup{colorlinks=true,       
	linkcolor=blueviolet,
	citecolor=oliveg,
}

\newtheorem{theorem}{Theorem}[section]
\theoremstyle{definition}
\newtheorem{definition}[theorem]{Definition}

\newtheorem{lemma}[theorem]{Lemma}

\newtheorem{corollary}[theorem]{Corollary}
\newtheorem{claim}[theorem]{Claim}

\newtheorem{remark}[theorem]{Remark}



\algrenewcommand\alglinenumber[1]{\footnotesize #1.}

\newcommand{\comment}[1]{}

\let\leq=\leqslant 
\let\geq=\geqslant %
\let\le=\leqslant 
\let\ge=\geqslant %

\newcommand{\paren}[1]{\left(#1\right)}

\newcommand{\brac}[1]{\left[#1\right]}

\newcommand{\set}[1]{\left\{#1\right\}}

\newcommand{\ceil}[1]{\left\lceil #1 \right\rceil}

\newcommand{\abs}[1]{\left\lvert#1\right\rvert}

\newcommand{\C}{\mathbb{C}}
\newcommand{\F}{\mathbb{F}}

\newcommand{\s}{\cS}

\newcommand{\R}{\mathbb{R}}

\newcommand{\Z}{\mathbb{Z}}
\newcommand{\Unif}{\textnormal{Unif}}

\newcommand{\cA}{\mathcal{A}}
\newcommand{\cB}{\mathcal{B}}

\newcommand{\cD}{\mathcal{D}}
\newcommand{\cE}{\mathcal{E}}
\newcommand{\cF}{\mathcal{F}}
\newcommand{\cG}{\mathcal{G}}

\newcommand{\cH}{\mathcal{H}}

\newcommand{\cL}{\mathcal{L}}
\newcommand{\cM}{\mathcal{M}}

\newcommand{\cS}{\mathcal{S}}
\newcommand{\CS}{\mathcal{S}}
\newcommand{\cT}{\mathcal{T}}

\newcommand{\innerprod}[1]{\left\langle#1\right\rangle}

\DeclareMathOperator*{\E}{\mathbb{E}}

\newcommand{\lcm}{\textnormal{lcm}}

\newcommand{\supp}{\textnormal{Supp}}
\newcommand{\colspace}{\textnormal{colspace}}

\newcommand{\juntadegtest}{\text{\sc Junta-deg}}
\newcommand{\weakdegtest}{\text{\sc Weak-deg}}
\newcommand{\degtest}{\text{\sc Deg}}
\newcommand{\Test}{\text{\sc Test}}

\makeatletter
\newcommand\footnoteref[1]{\protected@xdef\@thefnmark{\ref{#1}}\@footnotemark}
\makeatother

\definecolor{DSgray}{cmyk}{0,0,0,0.7}

\newcommand{\case}{}